\documentclass[10pt]{article}
\usepackage{cite}
\usepackage{subcaption}
\usepackage{amsmath, amsthm, amssymb,slashed,url}
\usepackage[misc]{ifsym}
\usepackage{euscript}
\usepackage{mathrsfs}
\usepackage{newfloat}
\DeclareFloatingEnvironment[fileext=lop]{Table}
\usepackage{caption}
\usepackage[scr=esstix, cal=boondox]{mathalfa}
\usepackage{titlesec}
 \usepackage{setspace} 
 \usepackage{xcolor}
 \usepackage{hyperref}
 \usepackage{cleveref}
 \usepackage{bbding}
\usepackage{caption, threeparttable}
\usepackage{blindtext,titlefoot}
\usepackage{ifpdf}
\ifpdf
  \usepackage[pdftex]{graphicx}
  \usepackage{epstopdf}
\usepackage{algpseudocode}
  \usepackage[]{algorithm2e}
 \usepackage{algorithmicx}
\else  
  \usepackage[dvips]{graphicx}
\fi
\usepackage{setspace}
\usepackage[utf8]{inputenc}
\usepackage{leading}
\leading{12pt}

\usepackage{amsthm}
\parskip=\baselineskip

\def\Tr{{\rm Tr}}
\def\16{{\bf 16}}
\def\1{{\bf 1}}
\def\2{{\bf 2}}
\def\3{{\bf 3}}
\def\4{{\bf 4}}

\def\rf{r_{\scriptscriptstyle{f}\, }}

\def\GWBI{{\sc GWBModel-I}}
\def\GWBII{{\sc GWBModel-II}}
\def\gammaR{\gamma_{\scriptscriptstyle{R}}}

\def\hsigma{\hat{\sigma}}

\def\tr{{\mathrm{tr}}}

\def\bp{\begin{pmatrix}}
\def\ep{\end{matrix}}
\def\BM{{\scriptscriptstyle{\text{\sc BM}}}}
\def\BL{{\scriptscriptstyle{BL}}}
\def\wBLI{\wgts_{{\text{\sc BL}_{\text{I}}}}}
\def\wBLII{\wgts_{{\text{\sc BL}_{\text{II}}}}}

\def\wGWI{\wgts_{{\text{\sc GWB}_{\text{I}}}}}
\def\wGWII{\wgts_{{\text{\sc GWB}_{\text{II}}}}}

\def\pinv{{+}}
\def\pdet{{+}}

\def\CCN{{\EuScript N}}

\def\CS{{\mathcal S}}
\def\hat{\widehat}
\def\CiX{\chi}
\def\conf{t}

\def\vecX{\vec{\chi}}
\def\vecY{\vec{\xi}}

\def\meas{\boldsymbol\rho}
\def\measF{\mathcal{F}}
\def\hcov{\hat{\cov}}
\def\sigmaF{\Xi}
\font\teneurm=eurm10 \font\seveneurm=eurm7 \font\fiveeurm=eurm5
\newfam\eurmfam
\textfont\eurmfam=\teneurm \scriptfont\eurmfam=\seveneurm
\scriptscriptfont\eurmfam=\fiveeurm

\font\teneusm=eusm10 \font\seveneusm=eusm7 \font\fiveeusm=eusm5
\newfam\eusmfam
\textfont\eusmfam=\teneusm \scriptfont\eusmfam=\seveneusm
\scriptscriptfont\eusmfam=\fiveeusm

\font\tencmmib=cmmib10 \skewchar\tencmmib='177
\font\sevencmmib=cmmib7 \skewchar\sevencmmib='177
\font\fivecmmib=cmmib5 \skewchar\fivecmmib='177
\newfam\cmmibfam
\textfont\cmmibfam=\tencmmib \scriptfont\cmmibfam=\sevencmmib
\scriptscriptfont\cmmibfam=\fivecmmib

\numberwithin{equation}{section}
\definecolor{brown}{rgb}{.5, 0.16, 0.16}
\definecolor{cyan}{rgb}{0.0, 0.85, 0.8}

\def\Uni{\EuScript{U}}

\def\normpdf{\phi}

\def\[{\left[}
\def\]{\right]}
\def\dr{{\text{\sc Drift}}}
\def\covar{{\text{\sc Cov}}}

\def\Agen{{Z}}

\def\BS{\EuScript{B} (\mathcal{s})}
\def\charv{\vec{{\mathcal{v}}}}
\def\ones{\vec{\mathscr{e}}}

 \def\sym{{\mathrm{sym}}}

\def\hvRet{{\hat{\vmu}_\Rsc}}
\def\CA{{\mathcal A}}

\def\tr{{\mathrm{tr}}}

\def\CO{{\mathcal O}}
\def\CP{\mathscr{P}}

\def\CL{{\mathcal L}}
\def\BLI{\text{\sc BL}_{\text{I}}}
\def\BLII{\text{\sc BL}_{\text{II}}}

\def\GWI{\text{\sc GWB}_{\text{I}}}
\def\GWII{\text{\sc GWB}_{\text{II}}}
\def\outperftstat{\mathcal{t}}
\def\scrS{\mathscr{S}}
\def\outperf{{\Delta \scrS}}
\renewcommand{\(}{\left(}
\renewcommand{\)}{\right)}
\def \half{{1\over2}}

\def\Tr{{\mathrm{Tr}}}

\def\vecZZ{\vec{Z}}

\def\subRet{{\scriptscriptstyle{\vec{R}}}}
\def\Ret{\vec{R}}
\def\Rsc{{{ R}}}
\def\Rsc{{\scriptscriptstyle{ R}}}
\def\vRet{\vec{\mu}_{\Rsc}}
\def\vdr{{\vec{\mathcal{z}}}}
\def\vdrV{{\vec{\mathcal{y}}}}
\def\vmu{\vec{\mu}}
\def\vnu{\vec{\nu}}
\def\vmm{\vec{\mathcal{m}}}
\def\sym{\displaystyle{\text{\normalfont Sym}}}

\def\cov{{\mathcal{C}}}

\def\scW{\mathscr{W}}

\def\prior{{\scriptscriptstyle{P}}}
\def\IPPT{{W}}
\def\viewV{\EuScript{V}}
\def\priorP{{P}}
\def\Est{{\scriptscriptstyle{E}}}
\def\viewVd{\EuScript{V}_\drift}
\def\viewVR{\EuScript{V}_\Rsc}
\def\view{{\scriptscriptstyle{\EuScript{V}}}}
\def\viewd{ {\scriptscriptstyle{\EuScript{V}_{\drift\, }}}}
\def\viewR{{\scriptscriptstyle{\EuScript{V}_{\Rsc\, }}}}
\def\viewP{{\CP}}

\def\cmu{{\mathcal{m}}}
\def\update{{\scriptscriptstyle{U}}}
\def\optU{{\star}}
\def\Diss{\text{\sc{Diss}}}
\def\eye{\mathbb{I}}
\def\imi{\mathscr{i}}
\def\ie{{\it{i.e.}}}
\def\Nprior{{N_a}}
\def\WW{\mathbb{I}_\Nprior + \lambda \viewP^T \viewP}
\def\Nview{{N_v}}

\def\High{\text{High}}
\def\Low{\text{Low}}
\def\BLmodelI{BL Model-I}
\def\BLmodelII{BL Model-II}
\def\mcv{\mathcal{v}}

\newcommand\gauDis[4]{{1\over \sqrt{(2\pi)^{#3} \det {#4}}}. e^{-\half\left( (\vdr - {#1})^T {#2} (\vdr - {#1}) \right)}}
\def\vmusim{{\vmu}_{\scriptscriptstyle{\text{\sc Sim},\wp}}}
\def\covsim{{\cov}_{\scriptscriptstyle{\text{\sc Sim},\wp}}}

\def \CM{{\mathcal M}}

\def\GBI{{\scriptscriptstyle{\text{\sc GWBI}}}}
\def\GBII{{\scriptscriptstyle{\text{\sc GWBII}}}}

\def\be{\begin{equation}}
\def\ee{\end{equation}}
\def\bea{\begin{eqnarray}}
\def\eea{\end{eqnarray}}
\def\wgts{\vec{\mathscr{w}}}
\def \wgx{\vec{\mathscr{x}}}

\def\pushP{{\viewP}_{{\text{\scalebox{.8}{$\sharp$}}}}}
\def\pushPi{{\viewP}_{{\text{\scalebox{.8}{$\sharp$}}}}^{(i)}}
\def\pushF{{\measF}_{{\text{\scalebox{.8}{$\sharp$}}}}}
\def\drift{d}
\def\covs{\cov(\mathcal{s})}

\newtheorem{lemma}{Lemma}[section]
\newtheorem{prop}{Proposition}[section]
\newtheorem{theorem}{Theorem}[section]
\newtheorem{remark}{Remark}[section]
\titleformat*{\subsubsection}{\it}
\usepackage{musicography}
\newtheorem{definition}{Definition}
\newcommand*\pushstar{\hspace{-.05cm}\includegraphics[scale=.275]{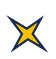}}
\newcommand*\updstar{\hspace{-.065cm}\includegraphics[scale=.275]{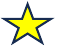}}

\makeatletter
\let\endeqno@equation\endequation
\def\endleqno@equation{\leqno \hbox{\@eqnnum}$$\@ignoretrue}
\def\lefteqnnum{\let\endequation\endleqno@equation}
\def\righteqnnum{\let\endequation\endeqno@equation}
\makeatother

\renewcommand{\footnoterule}{\vfill\kern -3pt \hrule width 0.4\columnwidth \kern 2.6pt}

\usepackage{tocloft}

\begin{document}

\begin{titlepage}
\begin{flushright}
\today
\end{flushright}
\vskip .350in
\begin{center}
{\bf\Large{A Geometric Approach To Asset Allocation With Investor Views}}
\vskip
0.25cm {\small Alexandre V. Antonov${}^{\spadesuit}$, Koushik Balasubramanian${}^{\spadesuit}$, Alexander Lipton${}^{\spadesuit,\dagger,\clubsuit,\star}$, Marcos Lopez de Prado${}^{\spadesuit, \dagger, \musNatural, \blacklozenge}$}
\vskip
0.1cm {\small${}^{\spadesuit}$Strategy and Planning Department, ADIA, Abu Dhabi, UAE} 
\vskip
0.1cm{\small${}^{\dagger}$ADIA Lab, Abu Dhabi, UAE} 
\vskip
0.1cm{\small${}^{\clubsuit}$Khalifa University, Abu Dhabi, UAE} 
\vskip 0.1cm 
{\small${}^{\star}$MIT Connection Science, MIT Cambridge,  MA, USA} 
\vskip 0.1cm 
{\small${}^{\musNatural}$ School of Engineering, Cornell University, NY, USA}
\vskip 0.1cm
{\small$\blacklozenge$ Lawrence Berkeley National Laboratory, CA, USA}
\end{center}
\vskip 0.2in

\baselineskip 10pt
\abstract{In this article, a geometric approach to incorporating investor views in portfolio construction is presented. In particular,  the proposed approach utilizes the notion of {\it generalized Wasserstein barycenter} (GWB) to combine the statistical information about asset returns with investor views to obtain an updated estimate of the asset drifts and covariance, which are then fed into a mean-variance optimizer as inputs.
Quantitative comparisons of the proposed geometric approach with the conventional Black-Litterman model (and a closely related variant) are presented. The proposed geometric approach provides investors with more flexibility in specifying their confidence in their views than conventional Black-Litterman model-based approaches. The geometric approach also rewards the investors more for making correct decisions than conventional BL based approaches. We provide empirical and theoretical justifications for our claim.
}

\vfill
\unmarkedfntext{$^*$ Author names are listed in alphabetical order of the last names. }
\end{titlepage}
\setcounter{tocdepth}{1}
\addtocontents{toc}{\protect\thispagestyle{empty}}
\tableofcontents
\thispagestyle{empty}
\newpage
\setcounter{page}{1}
\setstretch{1.0}

\section{\sc Introduction} \label{sec:Intro}
The Black-Litterman (BL) asset allocation model uses a Bayesian approach to infer the assets' expected returns based on a prior and views specific to investors \cite{BL}. Despite the vast amount of research work on this topic \cite{He, Bertsimas,Meucci,Rachev,Kolm,Meucci2014,Fabozzi, Duraj} the BL model continues to be an area of great interest. In the present article, we present a geometric approach to incorporate investor views rather than the conventional Bayesian approach used in the traditional BL model, as well as provide a means to incorporate their confidence in the views. Before we proceed to a formal introduction of the geometric approach, we will discuss the need for an alternative approach to incorporate views.

 To understand the need for an alternative approach it seems essential to understand the fact that the investor's personal confidence and precision of the views are independent. An investor who wishes to incorporate his or her views must provide expectation on asset returns along with the ``error-bars'' (or technically, ``confidence'' intervals) for the views. In fact, the investor needs to provide the complete information about the views distribution if the views are non-Gaussian. The ``confidence'' intervals do not represent the investor's personal confidence. An investor could choose to use the ``confidence'' intervals as a measure of personal confidence but he or she might choose to use other metrics (which could be subjective) for specifying their personal confidence. We will provide a concrete example to illuminate this remark $-$ if the investor believes that the methodology used to determine the views are not technically reliable then he or she will have no confidence in the views irrespective of the precision (or the ``error-bars''). For instance, the investor will have no confidence in a set of views if he or she discovers that the views were determined using look-ahead bias or corrupt data irrespective of the precision of the views.
 
 Though this is an extreme example, it demonstrates that the investor confidence and precision of the views are independent. To provide a less extreme example, let us consider an investor who uses proprietary signals to generate views systematically and also has views generated based on analysts estimates. Let us assume that the investor chooses to use only the proprietary model to determine the views on expected returns. The precision (inverse covariance) of the views could be computed from the historical predictions produced by the proprietary model. The investor's confidence in the proprietary-model-based views can be determined from the fraction of the observation period in which the proprietary model outperformed the model based on analysts estimates. In this example, it is again clear that the investor's confidence is unrelated to the precision of the views.
 
The conventional BL model incorporates the precision of views into the allocation process while the investor's subjective confidence is not incorporated. This claim will be demonstrated with the help of a {\it gedankenexperiment} in $\S$\ref{subsec:Gedanken}. For now, we will present some heuristic arguments to support this claim. An investor wishing to incorporate his or views should have the flexibility to specify any degree of confidence for a given views distribution.{\footnote{We refrain from using the term confidence level as this can be misinterpreted as the statistical confidence interval associated the views. The degree of confidence is the investor's subjective confidence on his or her views.}} That is, if the investor has 100\% confidence in his or her views then it is desirable to have the posterior or updated distribution match with the views distribution and if he or she has 0\% confidence then the desired update should match with the prior. For degrees of confidence strictly between 0\% and 100\%, then it is desirable to have the updated distribution smoothly interpolate between the prior and views. Figure \ref{fig:Confidence} shows the ``evolution'' of the {\bf desired} posterior distribution with the degree of confidence, for a hypothetical example where the prior and views distributions are Gaussian distributions on $\mathbb{R}^2$. In the conventional BL approach, if the prior and views distribution are specified then the prior and likelihood function for the Bayesian update rule are known and the posterior is computed from the product of the prior and likelihood function obtained from the views (see for {\it e.g.}, \cite{Kolm}). Hence, it is not possible to tune the investor's confidence in the conventional BL framework as it does not even appear in the update rule. From the earlier discussion, since the precision of the views and investor confidence are independent, it is clear that tweaking the parameters that change the precision of the views is not equivalent to tuning the investor's confidence. Hence, it seems that an alternate approach is needed for incorporating the subjective confidence of the investor in the allocation model.
 
 At first sight, a mathematical model to incorporate subjective confidence in an allocation model might seem infeasible. In this paper, we will describe a rigorous geometric approach to incorporate the subjective confidence of an investor. As observed earlier in the hypothetical example, confidence is a parameter that allows us to smoothly interpolate between the prior and views distribution. Interpolating between probability distributions is a well-studied topic in optimal transport theory. Optimal transport theory is a field of study that combines ideas from geometry and measure theory. In this paper, we propose an approach for incorporating investor's views by using the notion of generalized Wasserstein barycenter (GWB) introduced in \cite{GWB}. In particular, we show that the GWB of the prior and views distribution satisfies the desired properties of a posterior discussed earlier. We derive closed form expression for the GWB of the prior and views distribution which is a generalization of the McCann interpolant \cite{McCann}. This generalization is the main result of our paper.
 
Rest of this article is organized as follows: In $\S$\ref{sec:RevBL} we present a review of the original BL model and a closely related variant proposed by Meucci \cite{Meucci}. We notice that our geometric alternative based on the proposal in \cite{Meucci} has properties that are intuitive to an investor. Hence, it is worthwhile reviewing the proposal in \cite{Meucci} along with the original BL model. In $\S$\ref{subsec:Gedanken}, we present a {\it gedankenexperiment} to demonstrate that conventional BL models cannot interpolate between the prior and views distribution. In $\S$\ref{sec:GeomIntro}, we explain how the GWB is utilized in our geometric approach. In $\S$\ref{sec:GaussGWB}, we present the optimization problem for determining the GWB of the prior and views distribution. We also explain why the geometric approach can be extended to the case when the views are degenerate in $\S$\ref{subsec:Degen}. The main result of the paper is presented in $\S$\ref{sec:mainRes} where we present a closed form expression for the optimal update (or posterior) in our geometric approach. In $\S$\ref{sec:MVO} we show how the geometric updates can be used within the mean-variance optimization (MVO) framework. In sections $\S$\ref{sec:TestI} we provide methodologies for comparing the current approaches with the conventional BL approach (and its variant). Finally we conclude the paper by presenting a summary of our findings and presenting a brief outlook on the future directions.

\begin{figure}[t]
 \captionsetup{width=.85\linewidth}
\begin{center}
\hspace{.5cm}
\subfloat[]{\includegraphics[scale=0.2]{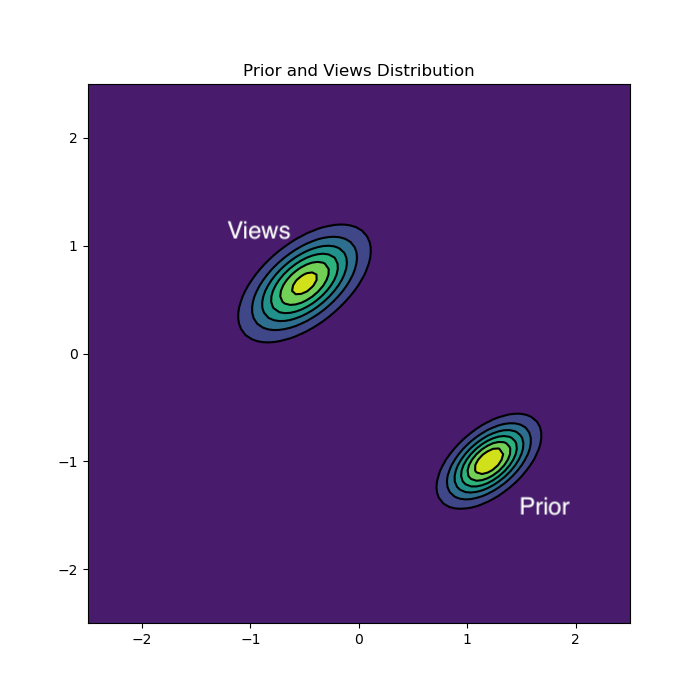}}
\subfloat[]{\includegraphics[scale=0.2]{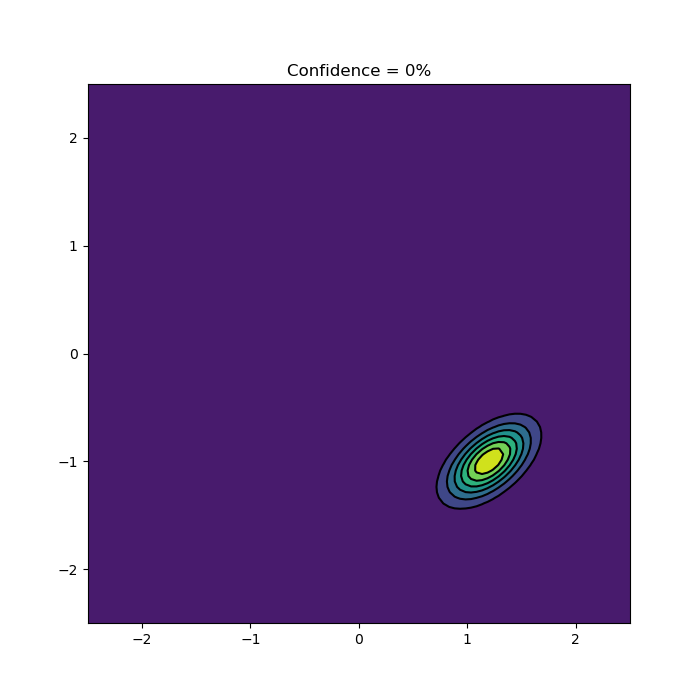}}
\subfloat[]{\includegraphics[scale=0.2]{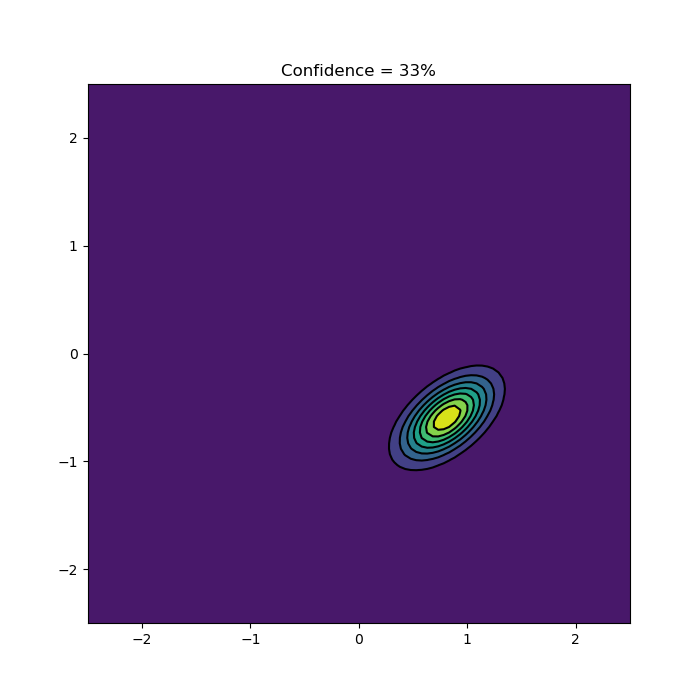}}
\end{center}
\begin{center}
\hspace{.5cm}
\subfloat[]{\includegraphics[scale=0.2]{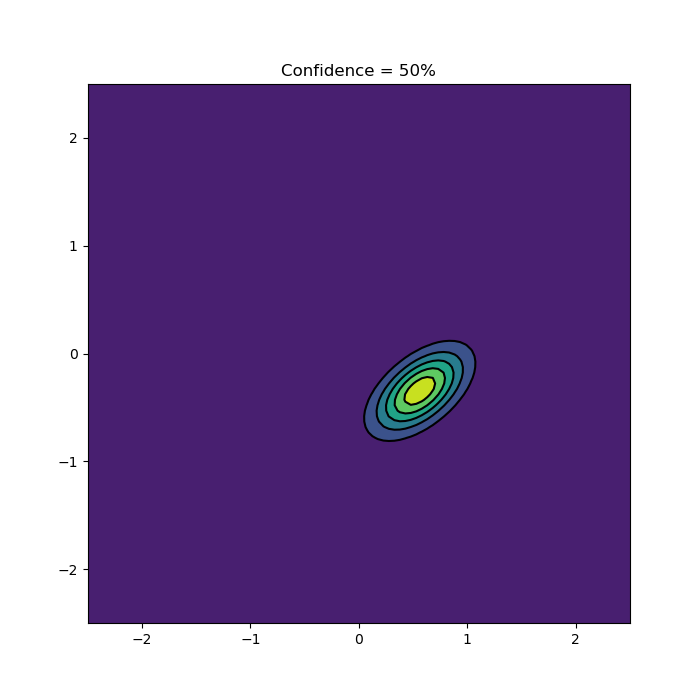}}
\subfloat[]{\includegraphics[scale=0.2]{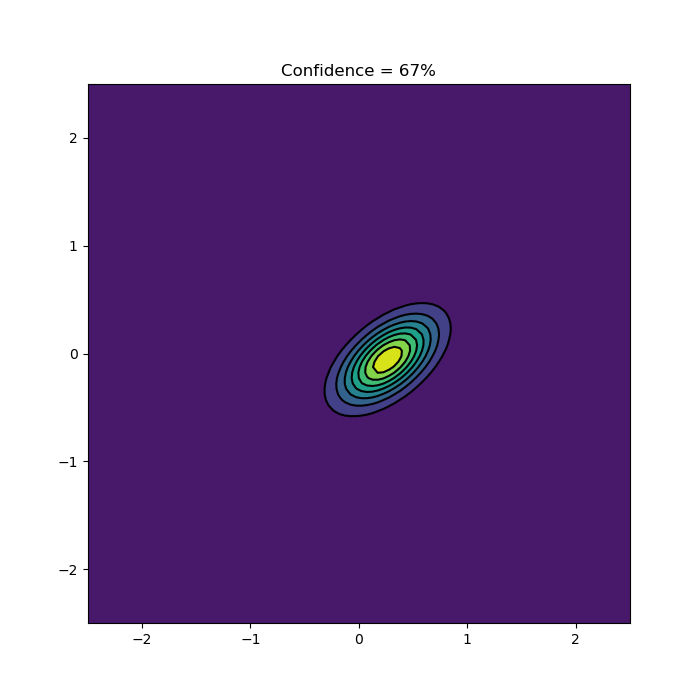}}
\subfloat[]{\includegraphics[scale=0.2]{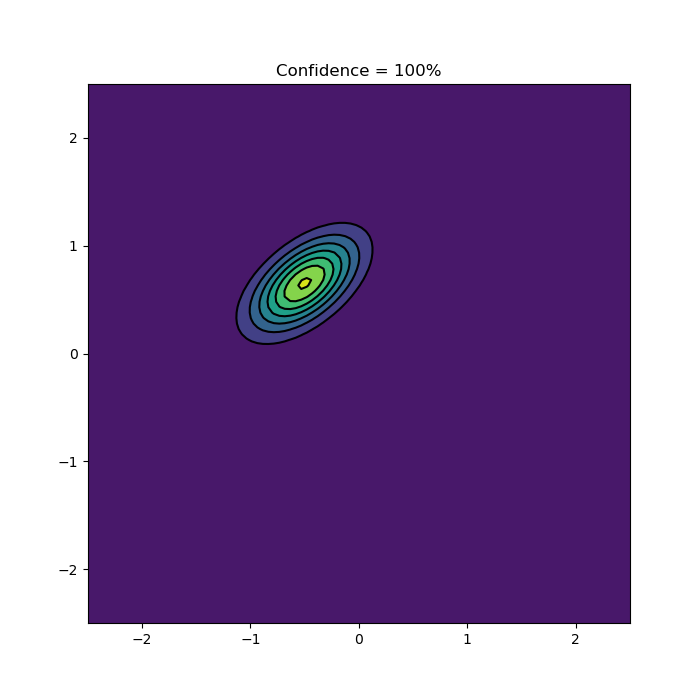}}
\end{center}
\begin{center}
\caption{(A) Shows a contour plot of hypothetical prior and views distributions. In this hypothetical example we assume there are only two assets and two views on the assets. We also assume that the distributions are normal as in the Black-Litterman model. (B)-(F) Show the desired updated distribution for different levels of investors' ``confidence''. When the investor is 100\% confident of his or her views, then it is desirable to have an updated distribution match with views distribution and when the confidence in the views is 0\%, then it is desirable to have an updated distribution match with the prior.\vspace{-.75cm}}
\label{fig:Confidence}
\end{center}
\end{figure}
\vspace{-.25cm}
\section{\sc Review of Black-Litterman Model {\it \&} a Variant}\label{sec:RevBL}

In this section, we present a lightning review of two versions of the Black-Litterman (BL) model. The review of BL model is in no way comprehensive and readers might find more elaborate reviews in the references (see for {\it e.g.}, \cite{BL} -\cite{Kolm}). In the first subsection, we will discuss the original proposal of Black and Litterman and in the second subsection we will discuss a variant proposed by Meucci. The two models differ in the way investors wish to incorporate their views. In the original BL model, the investors specify their views on the expected drift of a linear combination of assets (or drifts of certain portfolios). Subsequently, Meucci \cite{Meucci} proposed a minor modification of the model, where the prior beliefs and the investor views are directly specified on the asset returns instead of the drifts. In practice, these two approaches yield very different portfolios with different performance characteristics. 

 In this article we will refer to the conventional BL model or the original model proposed by Black and Litterman as \BLmodelI~and the variant discussed in \cite{Meucci} as \BLmodelII. 
We will now present a review of these two models.
\vspace{-.15cm}
\subsection{Original Black Litterman Model} \label{subsec:BLI}
 \vspace{-0.1cm}
A detailed discussion of the \BLmodelI~ will take us too far, however it is worthwhile reviewing the assumptions of the \BLmodelI \, and aspects of the model that are related to its underlying assumptions.
\begin{itemize}
\item {\bf Assumption 1}: The observable asset returns ($\Ret$) are assumed to follow a Gaussian distribution centered around a mean ($\vRet$) and the covariance of the returns is denoted by $\cov_\Rsc$.  Mathematically,
\be
\vec{R} \sim \CCN(\vRet, \cov_\Rsc), \qquad \vRet \in \mathbb{R}^{\Nprior} ,~ \cov_\Rsc \in {\text{Sym}}^{++}_\Nprior(\mathbb{R})
\ee
where $\vRet$ is the drift, $\cov_\Rsc$ is the covariance of returns and ${\text{Sym}}_N^{++}(\mathbb{R})$ is the set of all symmetric, real $N \times N$ positive definite matrices.  Though the assumption of Gaussianity of the asset returns is not completely corroborated by real-world data, it provides mathematical convenience and it is a relatively common assumption in the mathematical finance literature. Note that $\vRet$ and $\cov_\Rsc$ are unobserved quantities and need to be estimated. We will denote the estimate of $\vRet$ by $\hat{\vmu}_{\Rsc}$ and $\text{\sc Cov}(\Ret | \hvRet)$ by $\hat{\cov}_\Rsc$. In the original BL model, the estimate of $\hvRet$ is assumed to be uncertain and it is the next item in the list of assumptions. 

\item {\bf Assumption 2}: The estimate of the drift $\hvRet$ is assumed to be normally distributed with covariance ($ \cov_\drift$):
\be
\hvRet = \vmu_\drift + \vec{\epsilon}_\drift, \quad \text{where} ~ \vec{\epsilon}_\drift \sim \CCN(\vec{0}_\Nprior, \cov_\drift),  
\label{eq:BLIprior}
\ee
where $\vmu_\drift \in \mathbb{R}^\Nprior$ is the expected value of the estimated drift in returns, $\cov_\drift \in {\text{Sym}}^{++}_\Nprior(\mathbb{R})$ is the covariance of the estimated drift in returns, $\vec{0}_\Nprior$ denotes the zero-vector or the origin of $\mathbb{R}^\Nprior$ and $ \vec{\epsilon}_\drift$ models the noise resulting from the uncertainty in the estimation. Note that $\cov_\Rsc \neq \hat{\cov}_\Rsc \equiv \text{\sc Cov}(\Ret | \hvRet) $ as the uncertainties in $\hvRet$ also contribute to $\cov_\Rsc$  and in fact, $\cov_\Rsc = \hat{\cov}_\Rsc + \cov_\drift$ (see \cite{Kolm}, for instance).

 \hspace{0.5cm} The following example provides a simple approach for obtaining $\vmu_\drift$ and $\cov_\drift$ statistically. Historical sample mean is a simple {\it estimate} of the drift in the returns ($\hvRet$). Different estimates of the drift can be computed as the mean of multiple bootstrapped samples obtained by resampling the sample data. In this case, $\vmu_\drift$ is the bootstrap aggregated mean and $\cov_\drift$ is the bootstrap aggregation of the covariance of the drifts. However, this method of estimating the drift using historical returns cannot incorporate investor views and are often considered unsatisfactory to be used as the estimate for drift even in the absence of views  \cite{ BL}. Black {\it \&} Litterman \cite{ BL} provide an argument for estimating the drift $\vmu_\drift$ in the absence of investor specific views (that is, all investor views are identical). This argument will be discussed in the following assumption. 

\item {\bf Assumption 3}: 
If all the investors have identical views, then all investors positions align with the market (or a relevant benchmark portfolio) weights, $\vec{w}_{\BM}$. If all investors use an unconstrained mean-variance optimization with an average risk aversion parameter $\gammaR$ to determine the weights, then the 
 expected drift, $\vmu_\drift$ is obtained from the reference or benchmark weights ($\vec{w}_{\BM}$) by inverting the Markowitz optimality condition as shown below \cite{Bertsimas}
\be
\vmu_\drift = \rf \ones + \gammaR {\cov}_{\Rsc} \vec{w}_{\BM} = \rf \ones + \gammaR (\hat{\cov}_{\Rsc} + \cov_\drift) \vec{w}_{\BM}
\label{eq:BLRef}
\ee
where $\rf$ is the risk free rate. Equation (\ref{eq:BLRef}) is referred to as the ``equilibrium'' model as it explains the drift in asset returns when the market is in full-equilibrium where all participants have equal information and use the same methodology for allocation \cite{BL, Bertsimas}.
In general,  it is possible to obtain other estimates of the covariance matrix $ \hat{\cov}_{\Rsc}$ and the expected drift $\vmu_\drift$ through a reverse optimization procedure \cite{Bertsimas}, where the utility functions are different from the mean-variance based utility functions. It is also assumed that the covariance of $\hvRet $ is proportional to the conditional covariance of $\Ret$. That is,
\be
\cov_\drift = \tau \hat{\cov}_\Rsc
\ee
where $\tau$ is some scalar parameter, which has received a lot of attention from researchers \cite{Meucci}. It is worth noting that $0 \le \tau \le 1$ for $\vmu_\drift$ to be a reasonable estimate of $\hvRet$ or $\vmu_\Rsc$. This is because, the mean of expectation returns can be more accurately estimated than the mean of returns. If the equilibrium model drift is computed using the sample mean of an observation of length $T$, then we would have $\tau = 1/T$ (assuming independence). If $\tau$ is obtained based on a calibration procedure that compares the uncertainty of the equilibrium model with the sample estimator then it seems reasonable to set $\tau \approx 1/T$ \cite{Meucci, Rachev}. 
\item {\bf Assumption 4}: Investors and experts may have views ($\view_\drift$) that are not aligned with market (or the benchmark) and may wish to incorporate them in their allocation process. 
Note that the investor must also
provide a level of uncertainty by specifying $\cov_\viewd$. More generally, the investor specifies the views by specifying the distribution of expected returns, which could be non-normal. 
In the original BL model (\BLmodelI\,), the views distribution is assumed to be Gaussian.
That is,  the investors specify their views on the expected drifts (expectation on expected returns) of assets as shown below:
\be
\viewP. \hvRet  = \vnu_{\viewd}+ \vec{\eta}_{\viewd}, \quad \text{where} ~ \vec{\eta}_{\viewd} \sim \CCN(\vec{0}_\Nview, \cov_{\viewd})
\label{BLViewsd}
\ee 
where $\viewP \in \mathbb{R}^{\Nview \times \Nprior}$ is the views matrix which specifies the expected return on specific assets or some combinations of assets, 
$\vnu_{\viewd} \in \mathbb{R}^\Nview$, $\cov_\viewd \in \text{Sym}^{++}_\Nview(\mathbb{R})$ and $\Nview$
 is the number of views. Note that each row (denoted by $\vec{\mathscr{p}}_r$) of the views matrix $\viewP$ represents the weight of a portfolio $\Pi_r$ and the expectation of the expected return of this portfolio is $\nu_{\viewd,r}$ \cite{He}. The portfolio $\Pi_r$ could be a long only portfolio (even possibly with only asset) or could be a long-short portfolio. 
Note that $\viewP$ could be degenerate (in principle) due to the presence
  of multiple views (could even be conflicting) on the same assets. 
 If the views are independent, then the covariance matrix $\cov_\viewd$ associated with the views are diagonal. In a general case, it is possible to transform the views matrix $\viewP$ and the views drift $\vnu_\viewd$ in a way that makes $\cov_\view$ diagonal \cite{He}. 
However, for the purpose of this note, we will not make any assumptions about $\cov_\viewd$ and allow it to be a symmetric non-diagonal matrix. We have used the suffix $\viewVd$ for denoting the views covariance matrix, $\cov_\viewd$, to emphasize that the views are specified on the drifts. 
\end{itemize}
The BL model estimates the drift in the presence of views using a Bayesian approach where the prior distribution is given (\ref{eq:BLIprior}) with the drift parameter given by (\ref{eq:BLRef}) and the posterior distribution is obtained by computing the distribution of the expected returns given the views $\viewVd$ $\(\text{ denoted by }\mathbb{P}(\hat{\vmu}_{\Rsc}|{{\viewVd}}) \text{ in this article}\)$.  


We will now state the main result of the BL model:
Given the views ${\viewVd}$ on the drift in equation (\ref{BLViewsd}), the updated or posterior distribution of the estimated expected returns $\hvRet$ is given by
\be
\mathbb{P}(\hat{\vmu}_{\Rsc} |{\viewVd}) = \normpdf\(\hat{\vmu}_{\Rsc};\vmu_{\BL}, \cov^{(\vRet)}_{\BL}\)  
\label{BLdist}
\ee
where, $\normpdf\(\vecZZ ;\vmu, \cov\)$ is the probability distribution function of a Gaussian random variable, $\vecZZ$ $\sim$ $ \CCN\(\vmu, \cov\)$ and{\footnote{Using Woodbury identity $ \vmu_{\BL}$ and $\cov^{(\vRet)}_{\BL} $ can be written in a form that does not require the inverses of $\cov_\Rsc$ and $\cov_\viewd$ separately.}}
\be
 \vmu_{\BL} = \( \(\tau \hat{\cov}_\Rsc\)^{-1} +   \viewP^T \cov_\viewd^{-1} \viewP\)^{-1}\( \(\tau \hat{\cov}_\Rsc\)^{-1} \vmu_\drift+ \viewP^T\cov_\viewd^{-1} \vnu_{\viewd} \)
\label{BLmu}
\ee
\be
 \cov^{(\vRet)}_{\BL} = \( \(\tau \hat{\cov}_\Rsc\)^{-1} +   \viewP^T \cov_\viewd^{-1} \viewP\)^{-1}
 \label{BLcov}
\ee
The derivation of the above update equations can be found in \cite{Meucci, Kolm}. Note that the updated estimate for the distribution of asset returns is now given by,
\be
\mathbb{P}\({\Ret|\viewVd} \) = \normpdf\(\vmu_{\BL}, \hat{\cov}_{\subRet|\viewd} \),\quad {\text{where }}  \hat{\cov}_{\subRet|\viewd}  = \hat{\cov}_{\Rsc} +  \cov^{(\vRet)}_{\BL}
\label{BLCUpdate}
\ee
\subsection{A Variant of the Black-Litterman model}  \label{subsec:BLII}
 \vspace{-0.1cm}
In \cite{Meucci}, it was suggested that the expert views could be directly expressed on the raw asset returns instead of the expected returns. Meucci argues that specifying the views on the estimated drifts, as done in \BLmodelI, is often ``counterintuitive'' in limiting situations, even though the results are fully consistent. For instance, the covariance of the posterior distribution has a non-trivial dependence on $\tau$ even in the limit when the views are completely uninformative as well as the case when the views are completely correct. This dependence on $\tau$ stems from the fact that the estimated drift is uncertain which is inherent in the model assumptions.

Meucci  \cite{Meucci} proposed an alternate way of incorporating the views which have more intuitive limiting behaviors.
In this note, we will develop geometric methods that are analogous to both \BLmodelI\, and \BLmodelII\,  to check if any of the geometric methods yield ``counterintuitive'' result. Hence, it seems essential to understand the differences in the underlying assumptions of both these approaches.
\begin{itemize}
\item {\bf Assumption 1$'$}: As in the original Black-Litterman model, the observable asset returns ($\Ret$) are assumed to follow a Gaussian distribution centered around a mean ($\vRet$) and the covariance of the returns is denoted by $\cov_\Rsc$.  Mathematically,
\be
\vec{R} \sim \CCN(\vRet, \cov_\Rsc), \qquad \vRet \in \mathbb{R}^{\Nprior} ,~ \cov_\Rsc \in {\text{Sym}}^{++}_\Nprior(\mathbb{R})
\ee
Unlike the original BL model, it is assumed that $\vRet = \rf \ones + \gammaR {\cov}_{\Rsc} \vec{w}_{\BM}$ which eliminates the need for modeling it as a random variable.
\item {\bf Assumption 2$'$}: Expert views are expressed on the asset returns directly instead of the expected returns as shown below:
\be
\viewP. \vec{R}  = \vnu_{\view}+ \vec{\eta}_{\viewR}, \quad \text{where} ~ \vec{\eta}_\view \sim \CCN(\vec{0}_\Nview, \cov_\viewR)
\label{BLViews}
\ee 
\end{itemize}
We have used the suffix $\viewVR$ for denoting the views covariance matrix, $\cov_\viewR$, to emphasize that the views are specified on the asset returns directly. 
In this variant of the BL model, the distribution of the returns is updated as shown below:
\be
\mathbb{P}\(\vec{R}|\viewVR\) = \phi\(\vmu^{(\vec{R})}_{\BL'},  \cov^{(\vec{R})}_{\BL'}  \)
\label{BLMupdate}
\ee
where,
\be
 \vmu^{(\vec{R})}_{\BL'} = \( \hat{\cov}^{-1}_\Rsc +   \viewP^T \cov_\viewR^{-1} \viewP\)^{-1}\( \hat{\cov}^{-1}_\Rsc \hvRet + \viewP^T\cov_\viewR^{-1} \vnu_{\view} \)
\label{BLMmu}
\ee
\be
 \cov^{(\vec{R})}_{\BL'} = \(  \hat{\cov}^{-1}_\Rsc+   \viewP^T \cov_\viewR^{-1} \viewP\)^{-1}
 \label{BLMcov}
\ee
where $\hat{\cov}_\Rsc$ is an estimate for the covariance of returns ($\cov_\Rsc$) and $\hvRet$ is an estimate of the expected returns of the assets prior to incorporating any views. The above results can be obtained in the same manner in which the updates are derived in the Black-Litterman model. Note that the parameter $\tau$ does not appear in this model. The details of this derivation can be found in \cite{Meucci}. 

 \vspace{-0.25cm}
\subsection{A Simple {Gedankenexperiment}}\label{subsec:Gedanken}
 \vspace{-0.1cm}
Let us imagine that there is only one asset in the entire investible universe {\it i.e.}, $\Nprior=1$ in $\S$\ref{subsec:BLI} and $\S$\ref{subsec:BLII}. 
Let us also assume that the investor has a view about this asset ($\Nview =  1$) which could be a view on the expected returns of the asset (as in \BLmodelI) or the asset returns directly (as in \BLmodelII). 

First, we will present analysis of \BLmodelI\,.
We will denote the estimate of variance of the asset return ($R$) by $\hsigma^2_\Rsc$ and the variance of the expected returns of the asset ($\hat{\mu}_\Rsc$) by $\sigma^2_\drift$. In the notations of  $\S$\ref{subsec:BLI} we have, $\hcov_\Rsc = \hsigma^2_\Rsc = {\text{\sc Var}}(R|\hat{\mu}_\Rsc)$ and $\cov_\drift = \sigma^2_\drift = \tau  \hsigma^2_\Rsc$. The investor's view on the drift is denoted by $\nu_\view$ and the the corresponding variance is denoted by $\sigma^2_{\view}$. 
After incorporating the investor's view using \BLmodelI\, the updated expected return of the assets is given by,
\be
\mu_{\BL} = \({\sigma^2_\viewd\,  \hat{\mu}_\Rsc + \sigma^2_\drift\, \nu_\viewd \over \sigma^2_\viewd + \sigma^2_\drift} \)
 \equiv \({\sigma^2_\viewd\,  \hat{\mu}_\Rsc + \tau \hsigma^2_\Rsc\, \nu_\viewd \over \sigma^2_\viewd + \tau \hsigma^2_\Rsc} \)
 \label{eq:singleMu}
\ee
The above result is a direct application of equation (\ref{BLCUpdate}) for a single asset and a single view. Though the result in equation (\ref{eq:singleMu}) is sufficient for this {\it gedankenexperiment}, we will also present the updated variance of returns below (for the sake of completeness):
\be
\hsigma^2_{\subRet|\viewd}  = \hsigma^2_\Rsc + \sigma^2_{\BL}, \qquad \text{wher } \sigma^2_{\BL} = \tau \({ \hsigma^2_\Rsc \sigma^2_\viewd \over \tau \hsigma^2_\Rsc + \sigma^2_\viewd}\)
 \label{eq:singleVar}
\ee
Recall that $\CCN(\mu_{\BL}, \sigma^2_{\BL})$ is the updated distribution of the expected asset return (drift) while $\CCN(\mu_{\BL}, \hsigma^2_{\subRet|\viewd})$ is the updated distribution of the asset return. We will now present the main findings of this {\it gedankenexperiment}.

If the investor is completely confident about his or her views 
 then the intuitive expectation is that the posterior distribution will match with the investor's view distribution. 
However it is clear from equations (\ref{eq:singleMu}) and (\ref{eq:singleVar}) that \BLmodelI\, cannot produce the investor's view distribution as the updated distribution, for any value of $\tau$, since $\tau \in [0,1]$. Making the updated distribution align with the views distribution by choosing artificially high values of $\tau$ ($\tau \rightarrow \infty$) is illogical as it would imply $\sigma^2_\drift \gg \hsigma^2_\Rsc$, that is the noise in the expected returns is much greater than the noise in the returns. Therefore, it is not possible to obtain the views distribution as the posterior by tuning $\tau$ and hence $\tau$ is not a parameter that specifies an investor's personal confidence. It is in fact a parameter that specifies the ``error-bars'' for the estimates in expected return.
Since the investor provides views on the expected drift with the level of uncertainty, $\sigma^2_\viewd$ is not tunable either. In some research articles, $\sigma^2_\viewd$ by a $\tau$ dependent factor or $\tau$ could be defined in terms of the ratio of the number of observation points used for generating views to the number of observation points in the equilibrium model. In this case, taking $\tau \rightarrow \infty$ turns the view distribution into a Dirac delta. In summary, $\tau$ cannot be used to interpolate between the ``equilibrium'' distribution and the views distribution.

In the case of \BLmodelII\,,  it is clear from equations (\ref{BLMupdate})-(\ref{BLMcov}) that there are no tunable parameters. Thus it is not possible to obtain the investor's view distribution as the posterior without data dredging the views covariance matrix in \BLmodelII\, as well.

With the help of this simple {\it gedankenexperiment}, we have demonstrated that neither \BLmodelI\, nor \BLmodelII\, could reproduce the investor's views distribution as the posterior distribution without making illogical choices of parameters or data dregding. We will show that the geometric approach gives the investor flexibility to tune the confidence level so that the geometric posterior distribution will match with the investor's views distribution.
 
 \vspace{-0.25cm}
 \section{\sc{Distance Between Distributions}}\label{sec:GeomIntro}
 \vspace{-0.25cm}
 In this note, we will provide an alternate approach for incorporating investor views. In particular, we obtain the distribution of estimated drift (or returns) in the presence of views, as the generalized Wasserstein barycentre (GWB)  of
the views and reference distribution.  The focus of this section is to introduce the notion of GWB and discuss its relevance for asset allocation.

 In the following, the prior distribution could refer to the distribution of estimated drift or the distribution of the asset returns. 
  If the prior is assumed to be the estimated drift, then the views are expressed on the drift and in the other model, the views are directly expressed on the asset returns. We can then derive geometric methods that are analogous to \BLmodelI \, and \BLmodelII\, by a simple mapping and renaming of variables (discussed in remarks \ref{rem:BLI} and \ref{rem:BLII} of $\S$\ref{sec:mainRes}).

 We are interested in finding a target or updated distribution $f_\update$ that is as ``close" as possible to the prior (or reference) distribution, $f_{\prior}$, while staying in the ``proximity'' of the views.
``Proximity'' between distributions can be defined by introducing the notion of dissimilarity between distributions. The goal of the current approach can then be restated mathematically as follows:
\bea
\label{eq:OptUpdate}
(\cmu_{\optU},\cov_{\optU}) &= &\mathop{\text{argmin}}_{\cmu_\update, \cov_\update}  \Diss(f_\update , f_\prior) \\
\nonumber
\text{subject to} \qquad \qquad& &\\
\label{viewIneq}
\Diss(\pushP [f_\update] , f_\view)  &\le &d_0
\eea
where $\Diss(A,B)$ denotes a generic measure of dissimilarity between the distributions $A$ and $B$;  $\pushP [f_\update]$ denotes the {\it push-forward} of the ``update'' measure onto the views space along the map, $\viewP$. A formal definition of a push-forward measure can be found in Appendix \ref{app:push}. 

We can modify the optimization problem  in the constrained form to a Lagrangian form as shown below:
\bea
(\cmu_{\optU},\cov_{\optU}) &= &\mathop{\text{argmin}}_{\cmu_\update, \cov_\update} \left[ \(\Diss(f_\update , f_\prior)  + \lambda \Diss(\pushP [f_\update] , f_\view) \) \right]
\label{eq:LagForm}
\eea
where $\lambda$ is a Lagrange multiplier that plays the role of a tuning parameter that turns the constraint in (\ref{viewIneq}) into a term in the cost function. We would like to point out that dissimilarity or distance based approaches to BL models have appeared before in \cite{Meucci2014} and \cite{Duraj}.

 The optimization problem specified in equation ( \ref{eq:LagForm}) is in the Lagrangian form while the problem in equations ( \ref{eq:OptUpdate}) and (\ref{viewIneq}) is a constrained optimization problem (COP). The equivalence between the Lagrangian form and the COP form can be guaranteed by choosing a dissimilarity metric $\Diss$ such that Slater conditions are satisfied for all $d_0 > 0$ and $\Diss(A, B) \ge 0$ for any distributions $A$ and $B$ for which the dissimilarity is defined. 
Proposition \ref{prop:LagForm} in Appendix  \ref{app:LagForm} provides the precise details of the equivalence between the constrained optimization problem specified by equations (\ref{eq:OptUpdate})-(\ref{viewIneq}) and the optimization problem in equation (\ref{eq:LagForm}). 

 
The problem in (\ref{eq:LagForm}) is quite abstract since the dissimilarity measure has not been specified yet.
In this note, we will consider the Fr\'ech\'et or $L_2-$Wasserstein distance as the measure of dissimilarity. The definition of $L_2-$Wasserstein distance can be found in appendix B. The $L_2-$Wasserstein distance induces a metric on the space of probability measures. Note that the problem in equation (\ref{eq:LagForm}) can be written as the minimization of the following Lagrangian:
\be
\CL_{GWB} = \(\EuScript{D}_{WD}(f_\update , f_\prior)  +\lambda \EuScript{D}_{WD}(\pushP [f_\update] , f_\view)\)\ee
The above minimization problem is the same as computing the generalized Wasserstein barycenters (GWB) for two centers \cite{GWB} after writing $\lambda = t/(1-t)$ and multiplying $\CL_{GWB}$ by $(1-t)$ for $t \in [0,1)$. 
In \cite{GWB}, the authors consider the problem of finding the GWB when there are more than two centers. It is possible to obtain an analytical expression for the GWB of two Gaussian distributions and we will show that it is a generalization of McCann interpolant \cite{McCann}.  
The  problem in (\ref{eq:WDDist}) can be generalized to other distribution of returns and views. The views could also be prescribed through an arbitrary map $\viewP$ which need not be linear. 
However, we could only derive the analytical solution for the Gaussian case when $\viewP$ is linear. In other cases, the problem needs a numerical approach. 

In the next section we present the problem specialized to Gaussian distributions.

\begin{figure}[h]
\captionsetup{width=.85\linewidth}
\begin{center}
\includegraphics[scale=.25]{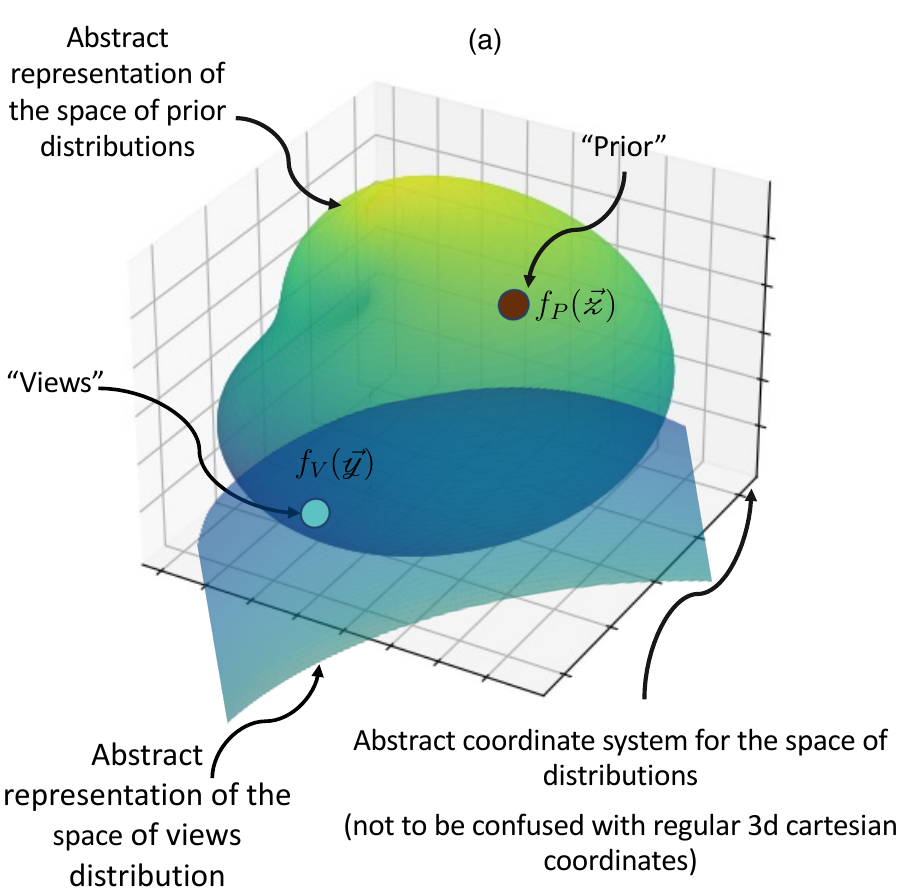}
\includegraphics[scale=.252]{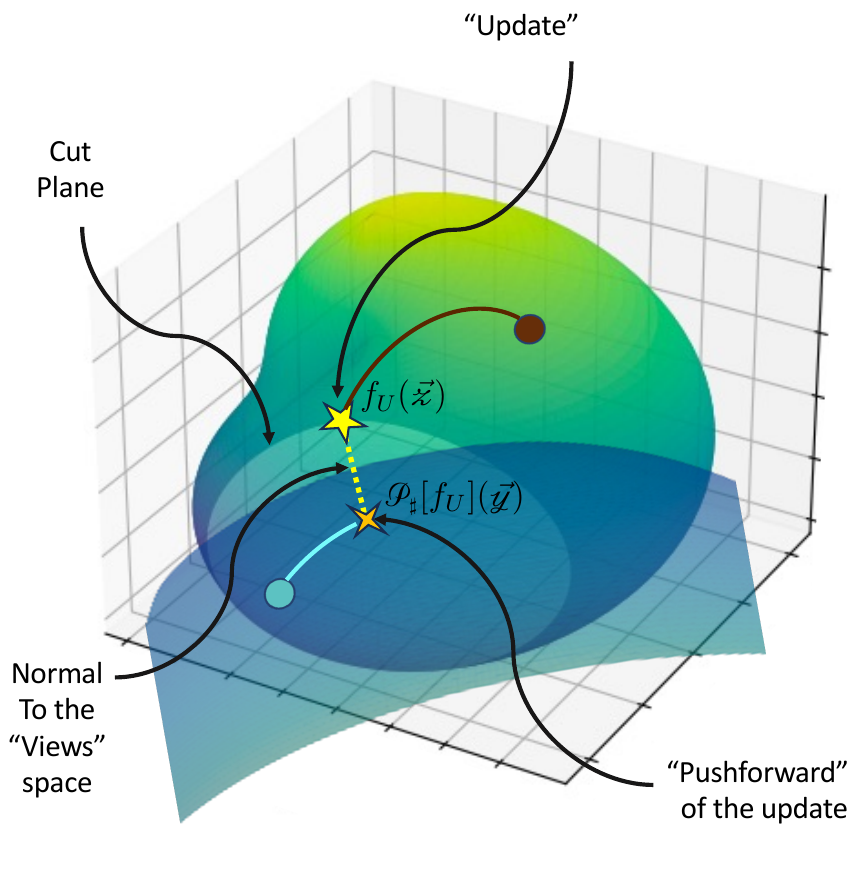}
\end{center}
\caption{{\bf (Left)} Shows an abstract representation of the space of probability measures containing the prior distribution $f_\prior$ and the space of measures containing the views distribution $f_\view$. In the space of probability measures, distributions are points, The point corresponding to $f_\prior$ (in the space of prior distribution) is represented by {\color{brown}\tiny \CircleSolid } (solid brown circle) and the point corresponding to $f_\view$ (in the space of views distribution) is represented by {\color{cyan}\tiny \CircleSolid } (solid cyan colored circle). {\bf (Right)} The {\it push-forward} of $f_\update$ on to the views space is denoted by \pushstar  (orange cross) and the update distribution ($f_\update$)  is represented by \updstar (yellow star).}
\end{figure}

 
 \section{{\sc GWB for Gaussian Prior {\it \&} Views}}\label{sec:GaussGWB}
 As mentioned earlier, the geometric method provides models that are analogous to \BLmodelI \, and \BLmodelII, which will be discussed in remarks \ref{rem:BLI} and \ref{rem:BLII} of $\S$\ref{sec:mainRes}. 
As in the BL model and its variants, we will assume that the prior and view distributions are Gaussian and have the following probability distribution functions (PDFs):
 \begin{gather}
f_\prior(\vdr) = {{1\over \sqrt{(2\pi)^{\Nprior} \det {\cov_\prior}}}. e^{ -\half(\vdr - \vmu_\prior)^T {\cov^{-1}_\prior}(\vdr - \vmu_\prior) }}, \\
f_\view(\vdrV) = {{1\over \sqrt{(2\pi)^{\Nview} \det {\cov_\view}}}. e^{ -\half(\vdrV - \vnu_{\view})^T {\cov^{-1}_\view}(\vdrV - \vnu_{\view}) }} 
\end{gather}
where  $\vdr,  \vmu_\prior  \in \mathbb{R}^\Nprior$, $\vdrV, \vnu_{\view} \in \mathbb{R}^\Nview$, $\cov_\prior \in \text{Sym}^{++}_\Nprior(\mathbb{R})$ and $\cov_\view \in \text{Sym}^{++}_\Nview(\mathbb{R})$. For convenience, we will refer to the subspace in which $\vdrV$ resides as the ``views'' subspace.
We will also assume that the target distribution $f_\update$ is Gaussian. Note that in the original BL model and in its variant, the updated distribution is Gaussian. 
Hence, we are justified in seeking a target or updated distribution that is also Gaussian.
\be
f_{\update}(\vdr) = \gauDis{\vmm_\update}{\cov^{-1}_\update}{N}{\cov_\update} 
\ee
where $\vmm_\update  \in \mathbb{R}^\Nprior$ and $\cov_\update \in \text{Sym}^{++}_\Nprior(\mathbb{R})$. To define the proximity to the views distribution it seems essential to define the distribution of $\viewP\vdr$ (which resides in the views subspace). 
However, the existence of such a distribution might be thwarted by the degeneracy of the views matrix $\viewP$ (for instance, identical rows in $\viewP$). 

\subsection{{Non-degenerate Views Matrix }}\label{sec:NonDegen}
Before we proceed to handle degeneracies in the view matrix, we will discuss the case where the distribution of $\viewP\vdr$ is well-defined and it is given by,
\be
{\pushP}[f_{\update}](\vdrV) =  {{1\over \sqrt{(2\pi)^{V} \det ({\viewP \cov_\update \viewP^T})}}. e^{ -\half(\vdrV - \viewP \vmm_\update)^T \(\viewP {\cov_\update} \viewP^T\)^{-1}(\vdrV - \viewP \vmm_\update) }} 
\label{eq:pushforward}
\ee
Push-forward of a Gaussian distribution along a linear map can be computed quite easily using the 
fact that $\viewP \vdr$ is a normal distribution. Hence, it is sufficient to compute $\mathbb{E}[\viewP \vdr]$ and  $\text{\sc Var}[\viewP \vdr]$. Note that $\mathbb{E}[\viewP \vdr]  = \viewP \vmm_\update$ and  $\text{\sc Var}[\viewP \vdr] = \viewP \cov_\prior \viewP^T$. We also provide a longer derivation of the result in equation (\ref{eq:pushforward}) using the formal definition of a push-forward measure in Appendix {\ref{app:push}. The computations in Appendix {\ref{app:push} can be extended to more general maps and distributions.

Note that when the views are degenerate the determinant in the denominator could vanish resulting in an ill-defined distribution. In this subsection we will assume that 
$\pushP[f_{\update}](\vdrV) $ exists, in which case it is possible to introduce notions of ``proximity'' between distributions. In the next subsection, a method for handling degenerate views will be presented. 
As mentioned  earlier, we will only discuss the case of non-degenerate views in this subsection. The $L_2-$Wasserstein distance between two Gaussian measures can be computed analytically (see for {\it e.g.}, \cite{WD,Dowson, Givens}). The details of the computation are presented in Appendix {\ref{app:WassDist}}.
Using equation (\ref{eq:appWassDist}) in Appendix {\ref{app:WassDist}} we get,
\begin{multline}
\label{eq:WDDist}
\qquad \qquad \CL_{GWB} = \| \vmm_\update - \vmu_\prior\|^2 + \tr\( \cov_\prior + \cov_\update - 2 \(\cov^\half_\prior \cov_\update \cov^\half_\prior \)^\half\) + \\
\lambda\(\| \viewP \vmm_\update - \vnu_{\view}\|^2 + \tr\( \cov_\view +\viewP \cov_\update\viewP^T - 2 \(\cov^\half_\view \viewP \cov_\update \viewP^T \cov^\half_\view \)^\half\)\)
\end{multline}
The expression in (\ref{eq:WDDist}) is well-defined even when the views matrix $\viewP$ is degenerate. Hence, the above cost function can be used for finding a target distribution that lies in the ``proximity'' of the prior and views, even when the views matrix is degenerate. The case of a degenerate distribution will be discussed in more details in a subsequent part of this note.  

\vspace{-.25cm}
\subsection{{Degenerate Views Matrix}}\label{subsec:Degen}

In this section, we will show how the geometric approach extends to the degenerate case. We will begin by starting with a formal definition of a multivariate normal (MVN) distribution and utilize this definition to the generalize the geometric approach to include degenerate views.
\begin{definition}
{\it{
A random vector $\vecX = \[ \CiX_1, \CiX_2, \dots \CiX_k\]^T$ has a multivariate normal distribution if $\vec{a}^T \vecX$ is a univariate random distribution for any $\vec{a} \in \mathbb{R}^k $. Note that a univariate normal distribution with zero variance is a Dirac delta distribution located at the mean of the distribution.
}}
\end{definition}
\vspace{-.25cm}
\noindent
The above definition is applicable even when the ``naive'' probability of $\vecX$ is degenerate \ie, when the covariance of $\vecX$ is not invertible. Alternatively, we could define the MVN distribution in terms of its characteristic function, $\varphi_{\vecX}(\vec{{\mathcal{v}}})$, of $\vecX$ as follows: 
A random vector $\vecX = \[ \CiX_1,\CiX_2, \dots \CiX_k\]^T$ has a multivariate normal distribution if the characteristic function, $\varphi_{\vecX}(\vec{{\mathcal{v}}})$, of $\vecX$ ,
\be
\varphi_{\vecX}(\vec{{\mathcal{v}}}) \equiv \mathbb{E}_{\vecX}\left[ e^{\imi \charv^T \vecX}  \right] = \exp \(\imi \charv^T \vmu  - \half \charv^T \cov \charv  \), \qquad \imi \equiv \sqrt{-1}
\label{eq:charfun}
\ee
for some $\vmu \in \mathbb{R}^k$ and $\cov \in \text{Sym}^{+}_k(\mathbb{R})$, where ${\text{Sym}}_k^{+}(\mathbb{R})$ is the set of all symmetric and real $k \times k$ positive {\it semi-definite} matrices. The probability mass function or the probability distribution function can be obtained as the Fourier transform of the characteristic function as shown below:
\be
\label{eq:GaussianChar}
\int {d^N \charv \over \( 2 \pi\)^N } \exp \(- \imi \charv^T \vecX + \imi \charv^T \vmu  - \half \charv^T \cov \charv  \) =  \({ 1\over  (2\pi)^{N} \det \cov}\)^{\half} \cdot \exp\( - {1\over 2} (\vecX - \vmu)^T \cov^{-1}\, (\vecX - \vmu) \)
\ee

In order to handle degenerate distributions, it is essential to define  pseudoinverse and pseudodeterminant of a matrix. We will show that a degenerate Gaussian distribution can be defined by replacing $\cov^{-1}$ in equation (\ref{eq:GaussianChar}) with the pseudoinverse of $\cov$ and $\det(\cov)$ with the pseudodeterminant of $\cov$, when the covariance matrix $\cov$ has zero eigenvalues.  
 The Moore-Penrose pseudoinverse (denoted by the superscript $\pinv$) and the pseudodeterminant (denoted by subscript $\pdet$) of a matrix $\Agen$ can be obtained using the following limiting procedure
\begin{gather}
\Agen^{\pinv} = \mathop{\lim}_{\delta \rightarrow 0} ( \Agen^T \Agen + \delta^2 \eye)^{-1} \Agen^T  = \mathop{\lim}_{\delta \rightarrow 0} \Agen^T ( \Agen \Agen^T + \delta^2 \eye)^{-1}  \\
{\det}_{\pdet}(\Agen) = \mathop{\lim}_{\delta \rightarrow 0} {1\over \delta^{2(N - \text{rank}(A))}} \det \( \Agen + \delta^2 \mathbb{I} \)
\end{gather}
Introducing a regularization parameter $\delta$ for the covariance in equation (\ref{eq:GaussianChar}), the normal distribution with a degenerate covariance can then be defined as follows:
\be
f_{\text{Degen}}(\vecX) = \({ 1 \over (2 \pi )^N {\det}_{\pdet}(\cov)}\)^{\half} \cdot \exp\( - {1\over 2} \(\vecX - \vmu\)^T \cov^{\pinv}\, \(\vecX - \vmu\) \)
\ee
The characteristic function of a degenerate distribution is still given by equation (\ref{eq:charfun}) which is well defined. In other words, the degenerate distribution can be defined as the inverse Fourier transform of the characteristic function (with appropriate regularization). The covariance can then be obtained by taking the second derivative of the characteristic function. To compute the Wasserstein distance between two Gaussian distributions it is sufficient that the second derivatives of characteristic functions of the two distributions are well-defined.
In Appendix \ref{app:WassDist}, we show that the Wasserstein distance is well-defined even when the covariance matrices of interest are degenerate.

In the next section, we will present the optimal updates for $\vmm_\update$ and $\cov_\update$.

\section{\sc{Main Result: Optimal Update}\label{sec:mainRes}}

\begin{theorem} \label{thm:Main}
 { $\CL_{GWB}$ is minimized when,
\be
\label{eq:optDrift}
\vmm_\update  = \vmm_{\optU} = \IPPT \( \vmu_\prior + {\lambda} \viewP^T  \vnu_{\view} \) \quad {\text{with }}~\IPPT = \(\mathbb{I}_{\Nprior} + {\lambda} \viewP^T \viewP\)^{-1}   = \IPPT^T
\ee
\be
\label{eq:covUpdateMain}
\cov_\update  = \cov_{\optU} =  \(\IPPT+ \EuScript{B}\)\cov_\prior\(\IPPT+ \EuScript{B}\)
\ee
where, $\EuScript{B} = \EuScript{B}^T $ and it is given by
\be
\EuScript{B} = \lambda \IPPT A^{-\half} \(A^{\half} \viewP^T \cov_\view \viewP A^{\half}\)^{\half} A^{-\half}  W , \qquad
A = \IPPT \cov_\prior \IPPT
\label{eq:geomCovA}
\ee
}
\end{theorem}
\vspace{-.5cm}
\begin{proof}
A detailed proof of this theorem is presented in Appendix \ref{app:OptUpdate}.
\end{proof}
\vspace{-.35cm}
\noindent
In equation (\ref{eq:geomCovA}),  $A^{\half}$ denotes matrix square root as usual and its existence is guaranteed by the spectral theorem. The result in Theorem (\ref{thm:Main}) is a generalization of the McCann interpolant for two Gaussian distributions living on (sub)spaces of different dimensions.
To our knowledge the result in Theorem (\ref{thm:Main}) and its proof in Appendix \ref{app:OptUpdate} have not appeared in the literature. Theorem \ref{thm:Main} is the main result of this article and in the following we present comments and checks of this result. 

\begin{remark}\normalfont
 The optimal update for the drift has no dependence on the prior or view covariance matrices.{\footnote{This expression for the drift update has a lot of resemblance to the drift update proposed by Doust \cite{Doust}. However, there are many crucial differences and the resemblance might just be a coincidence.}}  In particular, if $\Nview = \Nprior$ and $\viewP = \mathbb{I}_\Nprior$ (that is, the investor has absolute view on every single asset) then the update drift is simply a weighted average of $\vmu_\prior$ and $\vnu_{\view}$. 
\end{remark}
\begin{remark}\normalfont
In the case, when $\viewP^T \cov_\view \viewP$ is invertible, using Lemma \ref{lem:lemII} in Appendix \ref{sec:Lemma} repeatedly we get,
\bea
\label{eq:Asuka}
\cov_{\optU} &=& (\lambda W + \Gamma) \viewP^T \cov_\view \viewP (\lambda W + \Gamma), \qquad{\text{where }} \Gamma = W\cov^\half_\prior \(\cov^\half_\prior W \viewP^T \cov_\view \viewP W \cov^\half_\prior\)^{-\half}\cov^\half_\prior W\\
\label{eq:Bhatia}
{} &=& A + \lambda^2 W \viewP^T \cov_\view \viewP W +  \lambda (A  \viewP^T \cov_\view \viewP )^{\half} W +   \lambda W ( \viewP^T \cov_\view \viewP A)^{\half} 
\eea
To obtain the expression in equation (\ref{eq:Asuka}), we have used the definition of $\Gamma$ in equation (\ref{eq:Gamma}) of Appendix \ref{app:PropMin}.
When the views matrix $\viewP = \mathbb{I}$, equation (\ref{eq:Asuka}) reduces to the McCann interpolant (see Example 1.7 of \cite{McCann} or Lemma 2.3 of \cite{Asuka}) with the identification $t \rightarrow {1\over 1 + \lambda}$. Similarly equation (\ref{eq:Bhatia}) reduces to equation (39) and (63) in \cite{bhatia} when  $\viewP = \mathbb{I}$. This implies that $\cov_\optU$ is a point on the geodesic connecting the two points corresponding to $\cov_\prior$ and $\cov_\view$ on the Bures-Wasserstein maifold (when  $\viewP = \mathbb{I}$). 
The parameter $t$ controls the distance of $\cov_\optU$ from $\cov_\prior$ and in the financial context, the parameter $\lambda$ is used to control the confidence in the views of the investor. When the investor has complete confidence in the views then $\lambda \rightarrow \infty$; if the investor has very low confidence in the views then $\lambda \rightarrow 0$. 

In the case when the matrices $\cov_\view$ and $\cov_\prior$ are diagonal matrices and the views matrix $\viewP = \mathbb{I}$, we get the following simple expression for the updated volatility:
\be
\sigma_{\optU, i} = {\sigma_{\prior, i} + \lambda \sigma_{\view, i} \over 1 + \lambda} , \qquad {\text{where}}~\cov_{\circ} = {\text{\sc Diag}} \(\sigma^2_{\circ,i}\) \implies \sigma^2_{\circ, i} = \(\cov_{\circ}\)_{ii}, \quad \circ \in \{\optU, \priorP, \viewV \}
\ee
\end{remark}
\begin{remark}\label{rem:BLI}\normalfont
When an investor provides views on the expected returns, we set $\cov_\prior = \cov_\drift = \tau \hcov_\Rsc$,   $\cov_\view = \cov_\viewd$, $\vmu_\prior = \vmu_\drift$ and $\vnu_\view = \vnu_\viewd$. In this case, the updated distribution for the returns (in the geometric approach) is given by,
\be
\mathbb{P}_{\viewd, \optU}(\Ret) = \normpdf(\vmm_{\GBI}\,, \cov_\GBI) 
\ee
where, $\mathbb{P}_{\viewd, \optU}(\Ret) $ denotes the distribution of returns obtained from the optimal updates for the expected returns, $\vmm_{\GBI}$ and $\cov_\GBI$ are given by,
\bea
\label{eq:muGWBI}
\vmm_{\GBI} &=& \IPPT \( \vmu_\drift + {\lambda} \viewP^T  \vnu_{\viewd} \)\\
\label{eq:covGWBI}
 \cov_\GBI &=& \hcov_\Rsc + \tau\(\IPPT+ \EuScript{B}_{\viewd}\)\hcov_\Rsc\(\IPPT+ \EuScript{B}_{\viewd}\) \\
 \label{eq:BGWBI}
 \EuScript{B}_\viewd &= &\lambda \IPPT A_\drift^{-\half} \(A_\drift^{\half} \viewP^T \cov_\viewd \viewP A_\drift^{\half}\)^{\half} A_\drift^{-\half}  W\\
  \label{eq:AGWBI}
 A_\drift &=& \IPPT \cov_\drift \IPPT = \tau \IPPT \hcov_\Rsc \IPPT
  \eea
 Note that $\mathbb{P}_{\viewd, \optU}(\Ret)$ is not a conditional distribution.
We will refer to the model that uses the geometric approach to incorporate views on the {\bf expected returns} as the {\GWBI}\,.
\end{remark}
\begin{remark}\label{rem:BLII}\normalfont
When an investor provides views on the asset returns (as in \BLmodelII\,), we set $\cov_\prior = \hcov_\Rsc$,   $\cov_\view = \cov_\viewR$, 
$\vmu_\prior = \hvRet$ and $\vnu_\view = \vnu_\viewR$. In this case, the updated distribution for the returns (in the geometric approach) is given by,
\be
\mathbb{P}_{\viewR, \optU}(\Ret) = \normpdf(\vmm_{\GBII}\,, \cov_\GBII) 
\ee
where, $\mathbb{P}_{\viewR, \optU}(\Ret) $ denotes the distribution of returns obtained from the optimal updates for the expected returns, $\vmm_{\GBII}$ and $\cov_\GBII$ are given by,
\bea
\label{eq:muGWBII}
\vmm_{\GBII} &=& \IPPT \( \hvRet + {\lambda} \viewP^T  \vnu_{\viewR} \)\\
\label{eq:covGWBII}
 \cov_\GBII &=& \(\IPPT+ \EuScript{B}_{\viewR}\)\hcov_\Rsc\(\IPPT+ \EuScript{B}_{\viewR}\) \\
  \label{eq:BGWB|I}
 \EuScript{B}_\viewR &= &\lambda \IPPT A_\Rsc^{-\half} \(A_\Rsc^{\half} \viewP^T \cov_\viewR \viewP A_\Rsc^{\half}\)^{\half} A_\Rsc^{-\half}  W\\
   \label{eq:AGWBII}
 A_\Rsc &=& \IPPT \hcov_\Rsc \IPPT 
  \eea
We will refer to the model that uses the geometric approach to incorporate views on the {\bf asset returns} as the {\GWBII}\,.
\end{remark}
\begin{remark}\normalfont
When $\lambda =0$, we have $(\vmm_\optU, \cov_\optU) = (\vmu_\prior, \cov_\prior)$ and when $\lambda \rightarrow \infty$, we have $\viewP \vmm_\optU = \vnu_{\view}$ and $\viewP \cov_\update \viewP^T = \cov_\view$. Hence, $t = 1/(1+\lambda)$ plays the role of investor confidence, as it allows us to interpolate smoothly between the prior and views distribution (as described in $\S$\ref{sec:Intro}). The parameter $\lambda$ has no counterpart in the conventional BL model.
Note that in \GWBI which is the geometric analog of \BLmodelI\,, the updated drift of the returns align with the views drift when $\lambda \rightarrow \infty$, however the the updated covariance of returns is $\cov_\GBI = \hcov_\Rsc +  \viewP^T \cov_\viewd \viewP$ which depends on $\hcov_\Rsc$. This is counterintuitive as the updates $\vmm_\optU$ and $\cov_\optU$ match the views distribution . This is an artifact of the model that stems from the fact that the views are specified on the expected returns and not on the returns itself. This is also a feature of \BLmodelI\, which was pointed out by Meucci in \cite{Meucci}. The fact that the posterior drift of \GWBI\, match the views drift when the investor confidence is 100\% is a desirable feature - an investor is rewarded for being confident on correct views. 
In \GWBII\, which is the geometric analog of \BLmodelII\, the updated drift and covariance of the returns match the views distribution as the views are expressed directly on the returns. As explained in $\S$\ref{subsec:Gedanken}, neither \BLmodelI\, nor \BLmodelII\, can produce a posterior distribution that matches with the views distribution when the investor is 100\% confident on his or her views.

\end{remark}
\begin{remark}\normalfont  
An additional point worth mentioning is that the inverse of the update covariance matrix does not involve inverting $\cov_\prior$ and can also be written as follows:
\be
\cov^{-1}_\optU = W^{-1} A^{\half} \(A^{\half} W^{-1} A^{\half}  + \lambda \(A^{\half} \viewP^T \cov_\view \viewP A^{\half}\)^{\half} \)^{-2} A^{\half}  W^{-1}
\ee
It would be interesting to check if portfolios constructed using the above estimation for covariance matrix and drift are less sensitive to estimation errors. This note will not address questions surrounding sensitivity of portfolios constructed using the approach described here. We will however present an approach for comparing portfolios constructed using the approach described here and the traditional Black-Litterman approach. In the next section, we will briefly describe the portfolio construction methodology.
\end{remark}

\section{\sc{Incorporating Investor Views in Mean Variance Portfolio}\label{sec:MVO}}
In this note, we will compare the efficacy of incorporating investor views in
 a simple allocation model where the weights are computed by solving the following mean-variance optimization (MVO) problem:
\bea
\nonumber
{\text{MVO}}[\vmm_\Est, \cov_\Est; \gammaR, \rf]: & &\\
\wgts &=& \mathop{\text{argmax}}_{\wgx} \[ \({{\vmm}_\Est} - \rf \ones \, \)^T \wgx - {\gammaR \over 2 } \wgx^T  {\cov}_\Est \wgx    \]\\
\label{eq:MVOwgts}
\nonumber
\text{subject to} & & {}\\
\ones^T \wgx &=& 1\\
\label{const:MVOwgtsI}
x_i &  \ge  & 0, \quad \forall~ i \in \{1,2,\dots, N_a\}
\label{const:MVOwgtsII}
\eea
In the optimization problem specified above, $\rf$ is the risk free rate, $ \ones^T = [1,1,\dots,1]_{\Nprior} \equiv \vec{\mathbf{1}}_{\Nprior}$, $\Nprior$ is the total number of assets, $\gammaR $ is a risk aversion parameter (positive), $\vmm_\Est$ is an estimate for the drift and $\cov_\Est$ is an estimate of the covariance matrix.    In the rest of this article, we will assume $\rf = 0$. The optimization problem, ${\text{MVO}}[\vmm_\Est, \cov_\Est; \gammaR] \equiv {\text{MVO}}[\vmm_\Est, \cov_\Est; \gammaR, \rf = 0] $, is solved using {\sc CvxPy} \cite{cvxpy}.

 We have not included the impacts of transaction costs, holding or borrowing costs, slippage, {\it etc.},
 in our present analysis as the primary goal of this study is to compare the efficacy of the drift and covariance corrections. In realistic investment (or trading) processes it is often essential to enforce constraints on factor exposures and other trading constraints. In \cite{Enzo}, the authors provide a formulation that is capable of incorporating realistic cost models, constraints that are convex and certain risk measures that are different from the risk metric considered in Markowitz's original proposal. The analysis of \cite{Enzo} can be extended to incorporate investor views using the updated covariance and drift. However, we will not present such a study here as it will take us too far from the objective of this paper.

 We evaluate the efficacy of the two approaches by solving ${\text{MVO}}[\vmm_\Est, \cov_\Est; \gammaR]$ for the following four methods of estimating the drift and covariance:

  \begin{itemize}
 \item[$\bullet$]  $\BLI$ {\bf Allocation Methodology}: In this methodology, the views are specified on the expected returns (or the drift in returns). 
 The reference or prior model specifies the distribution of the expected drift and the updates are computed using \BLmodelI. The drifts and covariance 
 appearing in the updated distribution
  in equation (\ref{BLCUpdate}) are used as inputs to the mean variance optimization problem {\sc MVO} specified in
  equations (\ref{eq:MVOwgts}) - (\ref{const:MVOwgtsII}).
 A description of the methodology can be found in Appendix \ \ref{app:BLImethod}.
 
 \item[$\bullet$] $\BLII$ {\bf Allocation Methodology}: In this methodology, the views are specified on the asset returns directly. 
 The prior model specifies the distribution of the asset returns and the updates are computed using \BLmodelII. The drifts and
  covariance 
 appearing in the updated distribution
  in equation (\ref{BLMupdate}) are used as inputs to the mean variance optimization problem {\sc MVO} specified in
  equations (\ref{eq:MVOwgts}) - (\ref{const:MVOwgtsII}).
 A description of the methodology can be found in Appendix  \ref{app:BLIImethod}.

\item[$\bullet$] $\GWI$ {\bf Allocation Methodology}: In this methodology, the views are specified on the asset returns directly. 
 The prior model specifies the distribution of the asset returns and the updates are computed using \GWBI. The drifts and
  covariance 
 appearing in the updated distribution
  in equations (\ref{eq:muGWBI}) and (\ref{eq:covGWBI} - \ref{eq:AGWBI}) are used as inputs to the mean variance optimization problem {\sc MVO} specified in
  equations (\ref{eq:MVOwgts}) - (\ref{const:MVOwgtsII}).
 A description of the methodology can be found in Appendix \ref{app:GWBImethod}.

\item[$\bullet$] $\GWII$ {\bf Allocation Methodology}: In this methodology, the views are specified on the asset returns directly. 
 The prior model specifies the distribution of the asset returns and the updates are computed using \GWBII. The drifts and
  covariance 
 appearing in the updated distribution
  in equations (\ref{eq:muGWBII}) and (\ref{eq:covGWBII} - \ref{eq:AGWBII}) are used as inputs to the mean variance optimization problem {\sc MVO} specified in
  equations (\ref{eq:MVOwgts}) - (\ref{const:MVOwgtsII}).
 A description of the methodology can be found in Appendix  \ref{app:GWBIImethod}
 
  \end{itemize}

 \section{{\sc Testing \& Evaluation Methodology}\label{sec:TestI}}
 
In this section, we present a methodology for comparing the efficacy of the various methods of incorporating views in asset allocation discussed earlier. 
The testing or evaluation methodology consists of the following two components:
\begin{itemize}
\item[(i)] An evaluation where the 
inputs to the allocation methodologies can be controlled. We use simulated data (Gaussian) for this test in order to respect the assumptions of the allocation methodologies. This stage of testing will be called preliminary evaluation as it is designed in such a way that the backtesting principles are violated. This violation is required at this stage of testing in order to generate controlled views as the inputs to the allocation methodologies. If a methodology fails this stage of testing, it implies that the methodology does not work as expected.
The precise details of the preliminary evaluation procedure will be discussed later in this section.
\item[(ii)] In the second stage, we use ``walk-forward'' backtesting to evaluate the allocation methodologies. Backtesting at best only estimates the efficacy of an investment strategy on ``one single realization'' of the process that describes the market dynamics. Making decisions purely based on the backtested results on a ``single realization'' leads to overfitted strategies \cite{MLDP}. 
 Backtesting on synthetic paths that capture the stylized facts in historical market data is a reasonable alternative. However, the methodology for generating synthetic data and evaluation of the quality of synthetic data requires caution. The topic of generating realistic synthetic data generation is interesting in its own right and unfortunately a detailed discussion on this topic is beyond the scope of this paper. 
 In the present paper, we will present a simpler alternative to reducing the risks of backtest overfitting. This alternative approach will be discussed in $\S$\ref{subsec:TestII}.
\end{itemize}
 We will now present the details of the two stages of our testing methodology. 
 
\subsection{{\sc Stage I Testing: Simulated Data}}
 The goal of this stage of evaluation is to compare the different allocation methodologies in the three situations when the views are (i) ``correct'' (ii) ``ambiguous'' (iii) ``incorrect''. In the following, we provide a brief explanation of these three situations and the motivation to evaluate the methodologies in these three situations:
 \begin{itemize}
 \item[(a)] {\it Correct Views}: We say a view is ``correct'' when it aligns with the future realization of the returns or expected returns. In real trading, it is highly unlikely there is an investor who is correct about his or views consistently throughout history.{\footnote{In other words, we believe no investor has a ``clairvoyant crystal ball'' or the existence of one. If clairvoyant crystal balls exist, the authors would be searching for one instead of writing this paper.}} However, for the purpose of the preliminary evaluation we are interested in testing if the proposed allocation methodology can outperform the conventional method if an investor uses ``consistently correct'' views with high confidence. As emphasized earlier, an ideal allocation methodology should give an investor the flexibility to incorporate his or views with the desired degree of subjective confidence. In addition, it is desirable to have a methodology that rewards the investor for choosing the right level of confidence on his or her correct views.
 \item[(b)] {\it Ambiguous views}: An ``ambiguous view'' is a view that is uncorrelated with the future realization of the returns or expected returns. Though, no investor intentionally picks ``ambiguous views'', the market can behave erratically making the views look ambiguous. An investor can make an informed decision about his or her confidence in a view, if an allocation methodology underperforms when the views are ambiguous in comparison to ``correct views''.{\footnote{For example, if an investor makes more money from lottery winnings rather than his or her investment decisions, then he or she might be tempted to invest in lottery tickets rather than his or her investment ideas.}}
  \item[(c)] {\it Incorrect Views}: A view is ``incorrect'' when the future realization of returns or expected returns are negatively aligned with the view. Again it is highly unlikely that an investor is incorrect consistently, however it is desirable to have an allocation methodology that can penalize more for having more confidence in incorrect views. For instance, let us consider an investor who wishes to calibrate the confidence parameter (associated with a set of views) using backtested results on simulated or synthetic data. If the allocation methodology underperforms more often when confidence associated with incorrect views is high, then the calibration (or ``hyperparameter tuning'') methodology is more likely to assign lower confidence to incorrect views. 
 \end{itemize}
 
 \noindent
 We have not yet specified the procedure for generating views that can be classified as correct, ambiguous or incorrect. The precise methodology for views generation used in our preliminary evaluation and other details of the testing procedure are described below:
 
 \begin{itemize}
 \item For the purpose of the preliminary evaluation we use simulated returns data. In particular, we generate multiple samples of ($N_\wp$) daily returns time series of length $T$ for $\Nprior$ assets as follows: For each $\wp \in \{1,2,\dots, N_\wp\}$, we sample $T$ independent identically distributed random variables from a multivariate normal distribution $\CCN\(\vmusim, \covsim\)$ where $\vmusim \in \mathbb{R}^{\Nprior}$, $\covsim \in \text{Sym}_{\Nprior}(\mathbb{R})$. Note that for each $\wp$, the daily return series is in the form of a panel data with $T$ rows and $\Nprior$ columns. Note that the $\wp^{th}$ return series can also be represented as a path in $\Nprior-$dimensional space and we will refer to such a path as $\Nprior-$path. Hence, each simulated returns time series is a single sample from the space of all $\Nprior-$paths and we generate $N_\wp$ samples. In the our testing methodology we choose $T$ to be more than ten years, the number of assets ($\Nprior$) to be fifty and $N_\wp \sim 250$.
 \item For each $\wp$, we use each of the allocation methodologies $\BLI$, $\BLII$, $\GWI$ and $\GWII$ to construct portfolios of the $\Nprior$ ``simulated'' assets. In the following we describe the inputs to the allocation methodologies and the rebalancing details:
 \begin{itemize}
 \item The portfolios are rebalanced every quarter. We would like to emphasize that the rebalancing procedure used for the preliminary evaluation is not a realistic rebalancing as the views gener ating methodology are artificially tuned to align or misalign with the realized returns in the future.
 \item The covariance matrix of the prior distribution
 is estimated using the historical data using a look-back window of length $\ell_b$ (six months) ending on the rebalance day.  The drift of the prior distribution, $\vmu_\prior$, is computed using the reference model in equation (\ref{eq:BLRef})  assuming that the benchmark weights are all equal and sum up to 1. That is,
 \be
 \vmu_\prior = \gammaR \cov_\prior \wgts_{{}_{\BM}}, \qquad \text{where,}~ \wgts_{{}_{\BM}} = {1\over N_a} \vec{\mathbf{e}}
 \label{eq:BenchBLWD}
 \ee
 In the above equations, $\vmu_\prior = \vmu_\drift$, $\cov_\prior = \cov_\drift$, $\vnu_\view = \vnu_\viewd$ and $\cov_\view = \cov_\viewd$ for the $\BLI$ and $\GWI$ allocation methodologies, while for $\BLII$ and $\GWII$ allocation methodologies $\vmu_\prior = \vmu_\Rsc$, $\cov_\prior = \cov_\Rsc$, $\vnu_\view = \vnu_\viewR$ and $\cov_\view = \cov_\viewR$.
\item For the purpose of testing we choose $\viewP = \eye_\Nprior$, however the analysis in the rest of the text holds good for any general views matrix $\viewP$.
 \item We will now discuss the views generating process. In the preliminary evaluation, we use a {\emph{forward looking window}} ($\mathbb{F}_W$) of length $\ell_f$, starting from the date of rebalance. In this paper, we set $\ell_f$ to three years. For each method, we conduct experiments with the three different types of views mentioned before $-$
 \begin{itemize}
 \item[(a)] {\it Correct (but ``blurred'') views}: 
As mentioned earlier, we say the views are correct when the investor views align with future returns. That is, expected return and covariance of the views match with the expected returns and covariance of the returns in $\mathbb{F}_W$. We can get unreasonably good results if we assume that the views are perfectly match with the future returns. Hence, we ``blur'' the views as to make the views align with the future only approximately by introducing some uncertainty. That is we sample the $\cov_\view$ from a Wishart distribution and $\vnu_\view$ from a multi-variate normal distribution as shown below:
\begin{gather}
 \vnu_\view \sim \CCN\(\viewP \vmu_{\prior,\mathbb{F}_W}, \cov_\view\), \\
  \cov_{\view} = \ell^{-1}_f{\mathfrak{S}},  \quad \text{where } \mathfrak{S} \sim {\EuScript{W}}\(\ell_f, \viewP \cov_{P,\mathbb{F}_W} \viewP^T\)
\end{gather}
 where $\vmu_{\prior,\mathbb{F}_W}$ and $\cov_{\prior,\mathbb{F}_W}$ are the drift and covariance of the prior distribution estimated from the forward looking window. Note that expected 
 $\nu_\view$ is $\viewP \vmu_{\prior,\mathbb{F}_W}$ and the expected covariance is $\viewP \cov_{P,\mathbb{F}_W} \viewP^T$. This ensures that the views are aligned with the future returns.
  \item[(b)] {\it Ambiguous view}: When the views are ambiguous, $\vnu_\view$ has no positive or negative alignment with the future returns. Hence, we model ambiguous views as shown below:
  \begin{gather}
\vnu_\view \sim \CCN\(\vec{0}_{\Nview}, \cov_\view\), \\
  \cov_{\view} = \ell^{-1}_f{\mathfrak{S}},  \quad \text{where } \mathfrak{S} \sim {\EuScript{W}}\(\ell_f, \viewP \cov_{P,\mathbb{F}_W} \viewP^T\)
 \end{gather}
   \item[(c)] {\it Incorrect (but ``blurred'') views}: Incorrect views are modeled like correct views except that the drifts are drift of the views are negative aligned with the future returns as shown below:
  \begin{gather}
 \vnu_\view \sim \CCN\(-\viewP \vmu_{\prior,\mathbb{F}_W}, \cov_\view\), \\
  \cov_{\view} = \ell^{-1}_f{\mathfrak{S}},  \quad \text{where } \mathfrak{S} \sim {\EuScript{W}}\(\ell_f, \viewP \cov_{P,\mathbb{F}_W} \viewP^T\)
\end{gather}
Note that the drift of the views are exactly the opposite of correct views. 
 \end{itemize}

  \item Using the above methodology for estimating for prior and views and equations  (\ref{eq:BL1muE}), (\ref{eq:BLIImuE}), (\ref{eq:GWImuE}) and (\ref{eq:GWIImuE}), we compute $\(\vmm^{\BLI}_\Est, \vmm^{\BLII}_\Est, \vmm^{\GWI}_\Est, \vmm^{\GWII}_\Est\)$. Similarly we compute $\(\cov^{\BLI}_\Est,\right.$$\left. \cov^{\BLII}_\Est\,,\right.$ $\left. \cov^{\GWI}_\Est,\right.$$\left. \cov^{\GWII}_\Est\)$ using equations  (\ref{eq:BL1covE}), (\ref{eq:BLIIcovE}), (\ref{eq:GWIcovE}) and (\ref{eq:GWIIcovE}) and the estimation for prior and views obtained using the methodology described in the earlier points.
  \item  We define {\bf back-validation} as the procedure for evaluating how a strategy would play out on historical data if the future information required for validating the strategy was made available.{\footnote{The purpose of this definition is to distinguish the first stage of our testing methodology from regular backtesting.}} For example, in our paper we are interested in playing out the strategy when we provide correct or incorrect views and it is not possible to determine the correctness of a view without using future information. It is preferable to use the back-validation procedure on synthetic or simulated data that respects the assumptions of the model underlying the strategy. 
  
 \hspace{0.5cm} Using the weights allocation procedure described in Appendix \ref{app:AllocMethod}, we ``back-validate'' the four methodologies to compute the portfolios' returns and performance characteristics. We use a quarterly rebalancing schedule for all the four allocation methodologies. 
      \end{itemize}
   
   \item[$\bullet$] For every path $\wp$, the Sharpe ratios $\scrS_{\BLI}(\wp)$, $\scrS_{\BLI}(\wp)$, $\scrS_{\GWI}(\wp)
   $, $\scrS_{\GWII}(\wp)$ are computed. We also compute the Sharpe ratio for the benchmark allocation methodology (specified by $\wgts_{{}_{\BM}}$). The Sharpe ratio of the benchmark is denoted by $\scrS_{\BM}(\wp)$.{\footnote{Recall that we have set the risk-free rate to zero.}}
     \item[$\bullet$] We measure two methodologies using Sharpe ratio as the evaluation metric. The outperformance metric $\outperf(A,B)$ is defined as the difference in the expected Sharpe ratios of methodology $A$ and $B$. More precisely, the  outperformance metric $\outperf(A,B)$ is
     \be
     \label{eq:Outperf}
     \outperf(A,B) = {\mathbb{E}_\wp\[ \scrS_A(\wp) - \scrS_B(\wp)\] }
     \ee
     We can also measure outperformance using the difference in other performance characteristics such as Sortino ratio, Calmar, ratio, Omega ratio, {\it etc.},  however in the present paper we will use the difference Sharpe ratio as the metric. If $\outperf(A,B)$ is statistically significant, then we can infer that $A$ outperforms $B$. The outperformance is considered statistically significant if the following test statistic is above a critical threshold $\outperftstat_c$:
     \be
     \label{eq:Outperftstat}
     \outperftstat(A,B) = N^{\half}_\wp {\mathbb{E}_\wp\[ \scrS_A(\wp) - \scrS_B(\wp)\] \over \sqrt{{\text{\sc Var}}_\wp\[ \scrS_A(\wp) - \scrS_B(\wp)\]}}
     \ee
     where $N_\wp$ is the number of paths.
 \end{itemize}
 
 \subsection{\sc{Results of Stage I Testing}}
 For the purpose of the numerical study we chose $\gammaR = 2.5$, $\ell_b = 125$, $\ell_f = 750$, $\tau = \ell^{-1}_b$, $\Nprior = 50$, $\Nview = \Nprior$, $T = 4000$ and $N_{\wp} = 250$.{\footnote{Note that $T=4000$ corresponds to around fifteen years of daily returns.}} We present the findings for two different values of the confidence parameter $\conf$ defined as follows:
  $$\conf = {\lambda \over 1+ \lambda} $$
Note that $0 \le \conf \le 1$. In principle, $\conf$ can be tuned dynamically or determined through a hyperparameter tuning methodology. 

In our analysis, we examine the results of the methodology for two different values of the confidence parameter: $\conf = 95\%$ for high confidence and $\conf = 5\%$ for low confidence.  We would like to re-emphasize that $\conf$ is the investor's subjective confidence and not the confidence interval determined by the covariance or precision.  We will denote the geometric allocation methodologies with $\conf = 95\%$ by $\GWI\text{ (High)}$ and $\GWII\text{ (High)}$. Similarly, we denote those with $\conf = 5\%$ by $\GWI\text{ (Low)}$ and $\GWII\text{ (Low)}$.

In the following, we present our findings of the preliminary evaluation in three scenarios when the views are (a) correct, (b) ambiguous and (c) incorrect. The outperformance metric $\outperf$ is used for comparing $\GWI$ and $\GWII$ allocation methodologies (with confidence parameters $\conf = 95\%$ and $\conf = 5\%$)  with the benchmark, $\BLI$ and $\BLII$ methodologies. We choose a threshold of $\outperftstat_{c} = 3.125$ for the test statistic $\outperftstat$. This value of $\outperftstat_{c}$ corresponds to a significance level or $p-$value threshold of $0.001$ with the $N_{\wp}-1$ as the degree of freedom .{\footnote{We are only interested in one-sided tail.}}
 \subsubsection{Performance With Correct Views}
 Figure (\ref{fig:HCCV}) shows the distribution of Sharpe ratios for the different allocation methodologies when the investor's views are correct.
 \begin{figure}[htpb]
  \captionsetup{width=.85\linewidth}
\begin{center}
\includegraphics[scale=0.65]{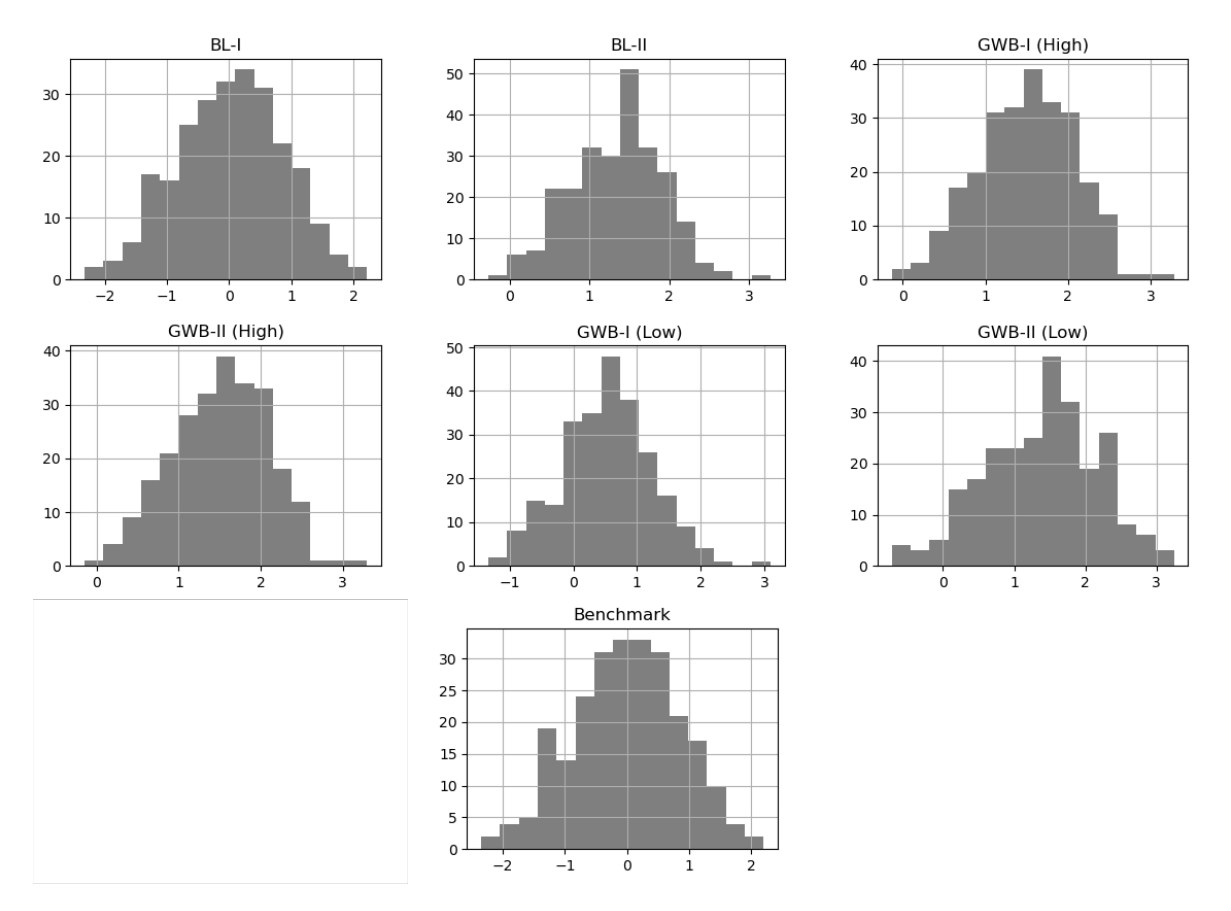} 
\caption{Shows the distribution of Sharpe ratios for the benchmark, $\BLI,\, \BLII,$ $\GWI,\, \GWII$ (High and Low) allocation methodologies.}\label{fig:HCCV}
\end{center}
\end{figure}
 It can be inferred from the location of the peaks of the histograms that the geometric approaches outperform the BL models. This can also be inferred quite directly from \autoref{table:Correct} (Top) which shows the outperformance metric $\outperf(A,B)$ for  $A \in \{\GWI,\, \GWII \}$ and $B \in \{\text{\sc BM},$ $\BLI,$ $\BLII,$ $\GWI,\, \GWII\}$.
 \begin{Table}[htpb] 
 \captionsetup{width=.85\linewidth}
 \vspace{-.25cm}
\begin{center}
\includegraphics[scale=0.75]{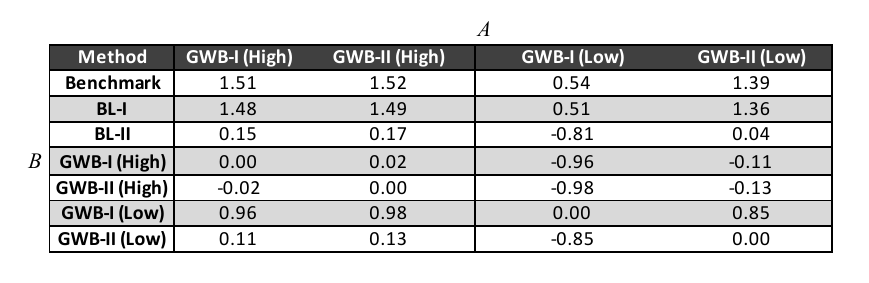}
\includegraphics[scale=0.75]{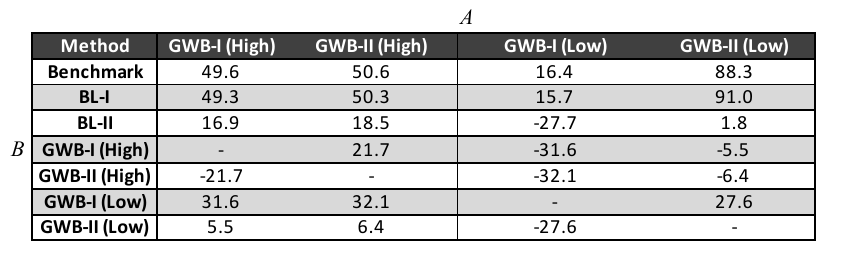}
\caption{{\bf(Top)} Shows the outperformance metric $\outperf(A,B)$ for  $A \in \{\GWI (\High),$ $\GWII (\High),$ $\GWI (\Low),$ $\GWII (\Low)\} $ and $B \in \{\text{\sc BM},$ $\BLI,$ $\BLII,$ $\GWI(\High),$ $\GWII(\High),$ $\GWI(\Low),$ $\GWII(\Low)\}$. {\bf(Bottom)} Shows the corresponding test statistic $\outperftstat(A,B)$. If  $\outperftstat(A,B)$ is lower than $\outperftstat_{c}$ is statistically insignificant. If $\outperftstat(B,A)$ is greater than $\outperftstat_{c}$ then the underperformance of $A$ compared to $B$ is statistically significant. }\label{table:Correct}
\end{center}
\vspace{-.50cm}
\end{Table}
Clearly both the geometric approaches based on GWB outperform the benchmark and both the Black-Litterman model when the views are ``correct'' and the investor has high confidence in the views.  The corresponding test static is shown in \autoref{table:Correct} (Bottom) and we conclude that the outperformance is significant. 

However, if the investor has low confidence on consistently ``correct views'', then he or she can only outperform the benchmark and $\BLI$ model using the geometric approaches.
This is because the geometric approaches are extremely close to the benchmark allocation methodology when the confidence is low. Since the $\BLII$ methodology clearly outperforms the benchmark (see Figure \ref{fig:HCCV}), it outperforms the geometric methods if the investor specifies low confidence. Interestingly, the $\GWII$ method also outperforms the $\GWI$ method irrespective of the degree of confidence.

 The geometric approach rewards the investor for having higher confidence in ``correct'' views. This seemingly ``qualitative'' statement is based on the empirical observation (from the top panel of \autoref{table:Correct}) that $\GWI(\High)$ and $\GWII (\High)$ outperform $\GWI (\Low)$ and $\GWII (\Low)$. This outperformance is statistically significant which is made clear in the bottom panel of \autoref{table:Correct}. In particular, we note that the test statistic $\outperftstat$ satisfies 
$$\outperftstat\(\GWI(\High),\GWI (\Low)\) > \outperftstat_c$$ 
$$\outperftstat\(\GWII(\High),\GWII (\Low)\) > \outperftstat_c$$ 
\subsubsection{Ambiguous View}
We now present the results for the case where the investor views are ambiguous. 
\begin{Table}[htbp]
\captionsetup{width=.85\linewidth}
\begin{center}
\includegraphics[scale=0.65]{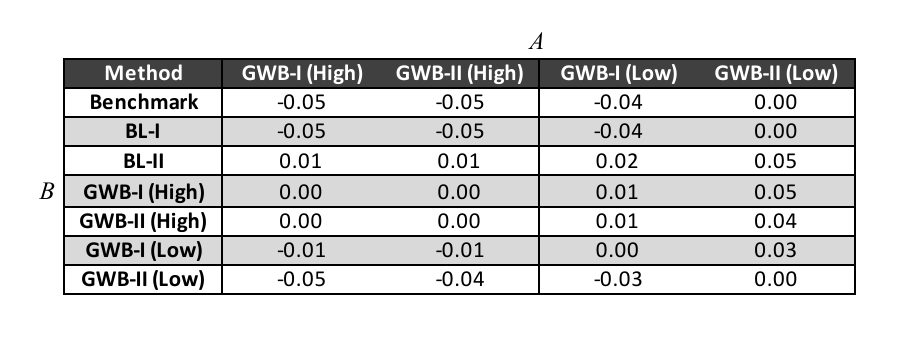}
\includegraphics[scale=0.65]{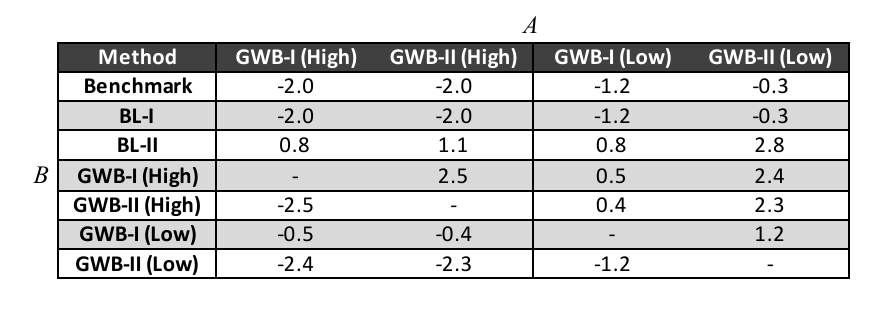}
\caption{{\bf{(Top)}} Shows the outperformance metric $\outperf(A,B)$ for  $A \in \{\GWI,\, \GWII\}$ and $B \in \{\text{\sc BM},$ $\BLI,$ $\BLII,$ $\GWI,\, \GWII\}$ and for $\conf = 95\%$ (High) and $\conf = 5\%$ (Low). {\bf{(Bottom)}} Shows the corresponding test statistic. If the test statistic $\outperftstat$ is less than $\outperftstat_{c}$, then the outperformance is statistically insignificant. It is clear from the values in the bottom table that the outperformance is statistically insignificant.}\label{table:Noisy}
\end{center}
\end{Table}
\vspace{-.1cm}
\autoref{table:Noisy} shows that when the views are ambiguous and have no relation to the future returns, the outperformance metrics are not statistically significant. This conclusion is expected because there should be no material outperformance (or underperformance) when views have no material information. This conclusion is independent of the degree of confidence as well.
 
 \subsubsection{Incorrect Views}

\autoref{table:Incorrect} shows that when if the investors provide consistently incorrect views with high confidence to the geometric approach, they underperform the benchmark as well as the Black-Litterman model. As explained at the beginning of this section, it is desirable to have a model that underperforms when the views are incorrect and when the confidence parameter is high. Recall that, if the confidence in the view is zero then the geometric model coincides with benchmark and in the presence of negative views, it is desirable to align with the benchmark. 
\begin{Table}[htbp]
\captionsetup{width=.85\linewidth}
\begin{center}
\includegraphics[scale=0.625]{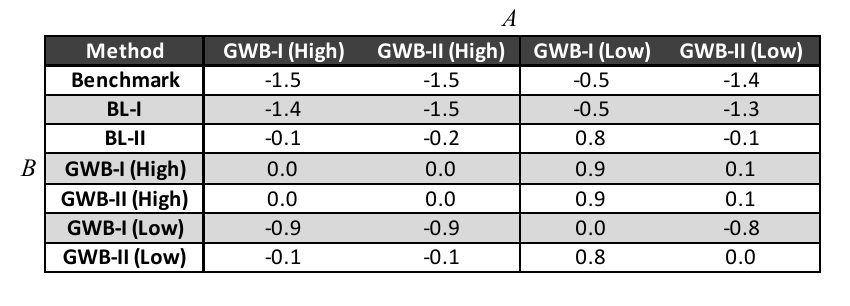}
\includegraphics[scale=0.625]{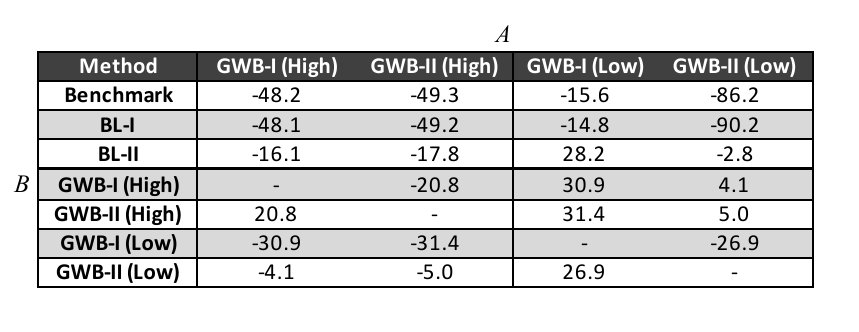}
\caption{{\bf{(Top)}} Shows the outperformance metric $\outperf(A,B)$ for  $A \in \{\GWI,\, \GWII\}$ and $B \in \{\text{\sc BM},$ $\BLI,$ $\BLII,$ $\GWI,\, \GWII\}$ and for $\conf = 95\%$ (High) and $\conf = 5\%$ (Low). {\bf{(Bottom)}} Shows the corresponding test statistic. If $\outperftstat(B,A)$ is greater than $\outperftstat_{c}$ then the underperformance of $A$ compared to $B$ is statistically significant.}\label{table:Incorrect}
\end{center}
\end{Table}
In the geometric approach, the investor is punished lesser for having lower confidence in ``incorrect'' views and in fact not punished if he or she has zero confidence in incorrect views. This observation can be inferred from the two panels in \autoref{table:Incorrect})  In particular, we note that 
$$\outperf\(\GWI(\Low),\GWI (\High)\) > 0, \qquad \outperftstat\(\GWI(\Low),\GWI (\High)\) > \outperftstat_c$$ 
$$\outperf\(\GWII (\Low),\GWII (\High)\) > 0, \qquad \outperftstat\(\GWII (\Low),\GWII (\High)\) > \outperftstat_c$$ 
 \subsubsection{Inference}
 
{\hspace{0.25cm}} From the results of our preliminary evaluation it is clear that the geometric approach behaves as desired in all the three situations when the views are (a) correct (b) ambiguous(c) incorrect. 
The confidence parameter which is absent in the conventional BL models provide additional flexibility to the investors, who can take advantage of this parameter and can (in principle) outperform the benchmark consistently with suitable judicious tuning of the confidence parameter. 

\hspace{-.25cm}The preliminary evaluation was based on unrealistic assumptions and ideal conditions which do not occur in real trading. Hence, it is essential to test the different allocation methodologies on real-world data. The procedure for testing with real data will be described in $\S$\ref{subsec:TestII}.

\subsection{{\sc Stage II Testing {\it \&} Results }\label{subsec:TestII}}
We will now discuss the second stage of the testing methodology which uses real data. We use the historical data of stock prices from {\sc Yahoo Finance} for this purpose. As mentioned earlier, the goal of the second stage of testing is to check how the allocation methodologies perform on real data. However, we do not have the luxury of testing on multiple ``paths'' as the real-world data is just one realization of the governing process. We will present an alternative approach to backtest on ``multiple paths'' which can reduce the risk of overfitting. The second stage of testing methodology is described below:
\begin{itemize}
\item To create multiple samples or ``multiple paths'' for backtesting, we choose $\Nprior$ assets out of a larger universe (denoted by $\Uni$) with $N_\Uni$ (greater than $\Nprior$) assets. This can be done in ${N_\Uni \choose \Nprior}$ ways. For sufficiently large $N_\Uni$ we get a large number of choices which all represent paths in $\Nprior-$dimensional space. Since each choice leads to an $\Nprior-$path, we can label a random selection of $\Nprior$ stocks by $\wp$ where $\wp \in \{1,2,\dots N_\wp\}$. Out of the ${N_\Uni \choose \Nprior}$ possible $\Nprior-$paths we choose $N_\wp$ paths and we will test our allocation methodologies using these $N_\wp$ samples (from real data).  For each $\wp \in \{1,2,\dots N_\wp\}$, we use each of the allocation methodologies $\BLI$, $\BLII$, $\GWI$ and $\GWII$ to construct portfolios of the $\Nprior$ chosen assets (labeled by $\wp$). In the following we describe the inputs to the allocation methodologies and the rebalancing details:
\item We assume that the prior or the reference distribution is determined using equation (\ref{eq:BenchBLWD}). That is, we assume that the benchmark weights are all equal. It is a common practice to use weights determined from the market capitalization as the benchmark weights. However, in a random selection of stocks, using market capitalization based weights could increase the risk of having concentrated benchmark weights. In this paper we will not analyze if capitalization based weights is a better choice for bench mark weights than equal weights. Interesting discussions on this topic can be found in the literature (for {\it e.g.}, \cite{mcap}), but the precise nature of benchmark weights is not crucial for our discussion.

Again we use, $\vmu_\prior = \vmu_\drift$, $\cov_\prior = \cov_\drift$, $\vnu_\view = \vnu_\viewd$ and $\cov_\view = \cov_\viewd$ for the $\BLI$ and $\GWI$ allocation methodologies, while for $\BLII$ and $\GWII$ allocation methodologies $\vmu_\prior = \vmu_\Rsc$, $\cov_\prior = \cov_\Rsc$, $\vnu_\view = \vnu_\viewR$ and $\cov_\view = \cov_\viewR$.

\item We will use the weights of a minimum variance (or volatility) portfolio for generating views as shown below:
\bea
\vec{\mu}_{\view} & = & \gammaR  \cov_\prior \wgts_{\text{\sc mVol}} \\
\vnu_{\view} & = & \viewP \vec{\mu}_\view
\eea
where  $\wgts_{\text{\sc mVol}} = \text{\sc MVO}[\vec{0}, \cov_\prior; \gammaR]$. Note that $\wgts_{\text{\sc mVol}} $ are weights of a long-only minimum volatility portfolio with weights adding up to 1.  
The views covariance matrix is obtained as described below:
\be
\cov_{\view} =  \viewP \cov_\prior \viewP^T
\ee
Note that the above estimates are obtained using the {\bf historical information only} and no forward looking information is used. The rationale behind using $\wgts_{\text{\sc mVol}}$ for generating views is to specify views that reduce the risk of the final portfolio. In the geometric approach, choosing a very high confidence on the views will ensure that the final portfolio lies in the proximity of a minimum volatility portfolio. Hence, we can hope to get portfolios that interpolate between equally weighted portfolio and a minimum volatility portfolio by tuning the confidence parameter.
  \item Using the above methodology for estimating for prior and views and equations  (\ref{eq:BL1muE}), (\ref{eq:BLIImuE}), (\ref{eq:GWImuE}) and (\ref{eq:GWIImuE}), we compute $\(\vmm^{\BLI}_\Est, \vmm^{\BLII}_\Est, \vmm^{\GWI}_\Est, \vmm^{\GWII}_\Est\)$. Similarly we compute the covariances as we did in the preliminary evaluation methodology,  using equations  (\ref{eq:BL1covE}), (\ref{eq:BLIIcovE}), (\ref{eq:GWIcovE}) and (\ref{eq:GWIIcovE}).
\item Using the weights allocation procedure described in Appendix \ref{app:AllocMethod}, we backtest the four methodologies to compute the portfolios' returns and performance characteristics for every choice of $\Nprior$ assets (that is for every $\wp \in \{1,2, \dots, N_\wp\}$). We use a quarterly rebalancing schedule for all the four allocation methodologies. We would like to emphasize that the ``walk-forward'' backtesting is used for this stage of the testing methodology and only historical information is used.
\item We then compute the outperformance metric (difference in Sharpe ratios) and the test statistic $\outperftstat$ using equations (\ref{eq:Outperf}) and (\ref{eq:Outperftstat}) as done in the first stage of testing.
\end{itemize}

 \subsection{\sc{Results of Stage II Testing}}

For the purpose of this study we choose the stocks which are the current constituents of S {\it \&} P 500 having around fifteen years of data as the universe $\Uni$. This has over 350 stocks out of which we choose $\Nprior = 50$ stocks at random. By ensuring the stocks have fifteen years of data we ensure that the universe size does not change with time. All the model parameters are the same as the ones used in the preliminary evaluation. As done in the first stage of testing, we will present the results for two different values of the confidence parameter:(i) $\conf = 95\%$ and (ii) $\conf = 5\%$. All variables and methodologies' names are the same as the ones used in the preliminary evaluation. We now present the result of our testing.

Figure (\ref{fig:real}) shows the distribution of Sharpe ratios for the different allocation methodologies when used on random selection of $\Nprior$ real assets. 
 \begin{figure}[htpb]
  \captionsetup{width=.85\linewidth}
\begin{center}
\includegraphics[scale=0.65]{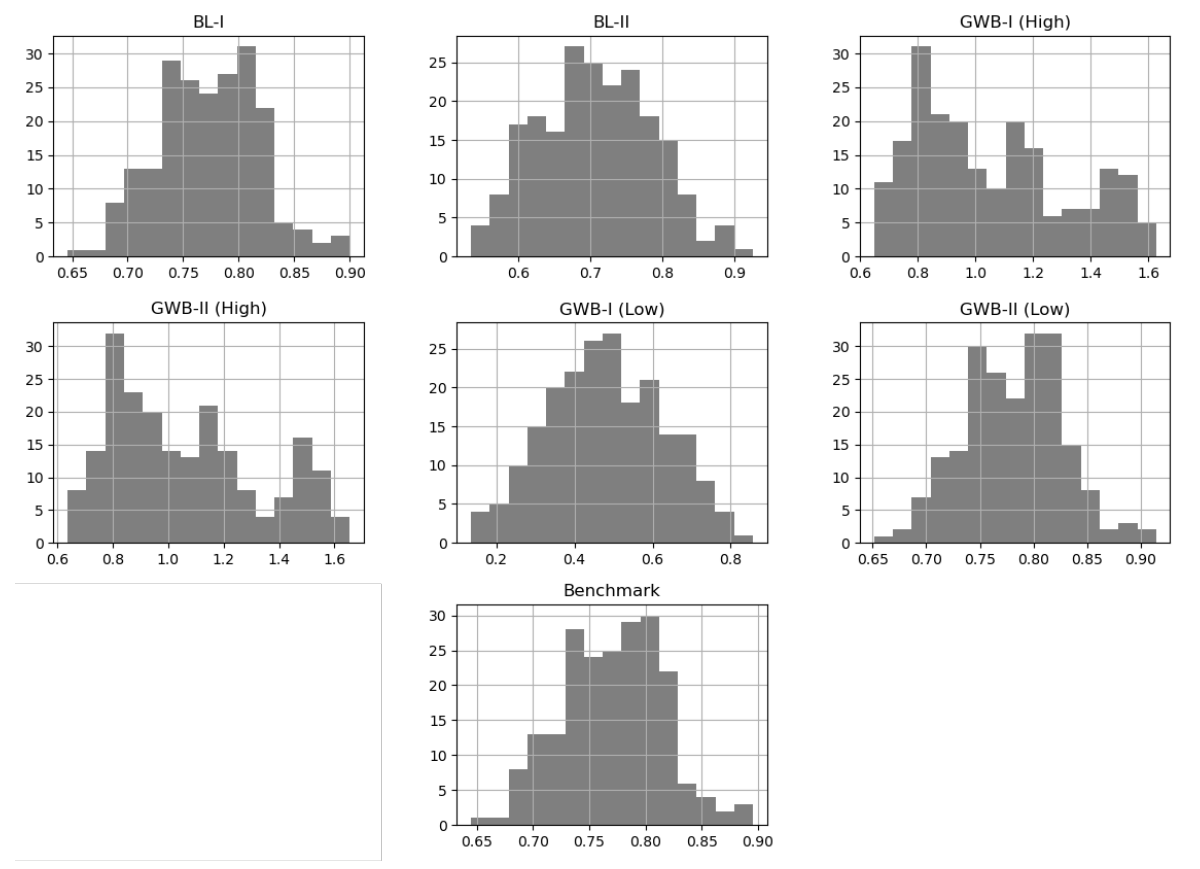} 
\caption{Shows the distribution of Sharpe ratios for the benchmark, $\BLI,\, \BLII,$ $\GWI,\, \GWII$ (High and Low) allocation methodologies.}\label{fig:real}
\end{center}
\end{figure}
It is quite evident from the histogram plots that the geometric approaches outperform the benchmark and the BL models when a high degree of confidence is specified for the views. 
 It is also clear from \autoref{table:Real} (Top) that the $\GWII$ approach performs far better than the conventional BL models and even the $\GWI$ model. 
 \begin{Table}[htpb] 
 \captionsetup{width=.85\linewidth}
 \vspace{-.25cm}
\begin{center}
\includegraphics[scale=0.75]{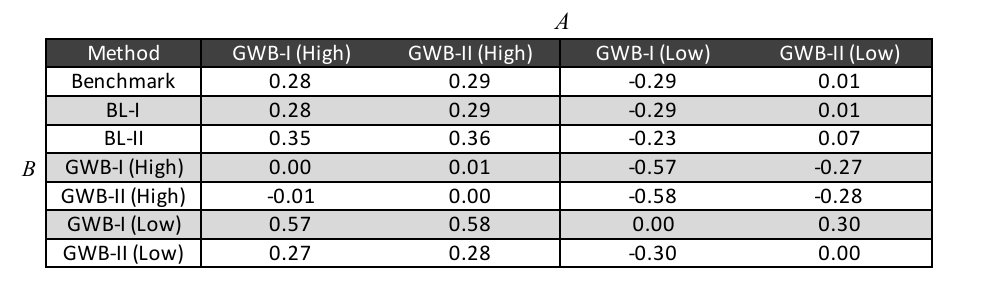}
\includegraphics[scale=0.75]{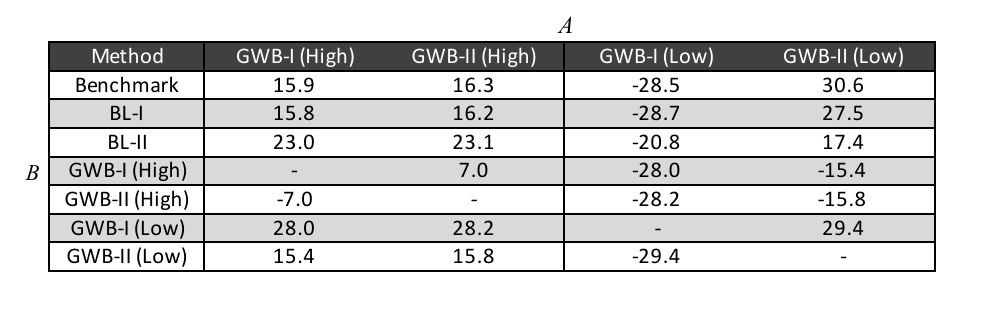}
\caption{{\bf(Top)} Shows the outperformance metric $\outperf(A,B)$ for  $A \in \{\GWI (\High),$ $\GWII (\High),$ $\GWI (\Low),$ $\GWII (\Low)\} $ and $B \in \{\text{\sc BM},$ $\BLI,$ $\BLII,$ $\GWI(\High),$ $\GWII(\High),$ $\GWI(\Low),$ $\GWII(\Low)\}$. {\bf(Bottom)} Shows the corresponding test statistic $\outperftstat(A,B)$. If  $\outperftstat(A,B)$ is lower than $\outperftstat_{c}$ is statistically insignificant. If $\outperftstat(B,A)$ is greater than $\outperftstat_{c}$ then the underperformance of $A$ compared to $B$ is statistically significant. }\label{table:Real}
\end{center}
\vspace{-.50cm}
\end{Table}
The $\GWI$ approach underperforms all the methodologies when a low confidence is specified. In our analysis, $\GWII$ method has outperformed all other approaches in both stages of testing and it is also intuitive. It would be worth exploring this allocation methodology in more detail.
\section{{\sc Conclusions {\it \&} Outlook}}
In this paper, we have presented a geometric approach to incorporating investor views that utilizes ideas from optimal transport theory. Given the growing number of applications of optimal transport theory in machine learning, computer vision, physics {\it etc.}, it is not surprising that optimal transport has utility for portfolio construction. The approach presented in this paper provides an investor the flexibility to specify the confidence in the form of a parameter that does not exist in the conventional BL models. We have provided empirical evidence and theoretical arguments to demonstrate that the geometric approach rewards skillful investors, who can adjust their confidence in their views judicially, more than the conventional BL models. 

From a systematic investing perspective, it would be interesting to build an allocation methodology that tunes the confidence parameter dynamically based on regime shift models that can identify if a view is correct, incorrect or ambiguous. An investor who wishes to incorporate different views with different levels of confidence can do so by using the multi-center GWB \cite{GWB} \ie, by solving the minimizing the following Lagrangian,
\be
\label{eq:concGWB}
\CL_{GWB} = \(\EuScript{D}_{WD}(f_\update , f_\prior)  +\mathop{\sum}_{i=1}^{K} \lambda_i \EuScript{D}_{WD}(\pushPi [f_\update] , f_\view)\)
\ee
where $\viewP^{(i)}$ denotes the views matrix for the $i^{th}$ view and $t_i = {\lambda_i\over 1 + \lambda_i}$ is the confidence associated with that view. The Lagrangian in equation  (\ref{eq:concGWB}) can in principle be minimized numerically, however the authors are not aware of a closed form expression for the GWB when the number centers ($K+1$) is more than two. Formally, the GWB problem in equation (\ref{eq:concGWB}) can also be extended to non-Gaussian distributions. 

Note that the main challenge in minimizing $\label{eq:concGWB}$ for Gaussian distributions is deriving the covariance update. The covariance update rule can be used for forecasting covariance matrices and it could use multiple methods of estimating covariance as views. For example, the covariance update rule (in equation \ref{eq:covUpdateMain}) can be used for finding the barycenter of factor model covariance and historical covariance. 

We believe that the geometric approach presented here has many interesting applications in finance and the methodology presented here will provide uncorrelated approaches to incorporating investor views. 

\noindent
\begin{center}
{\bf Acknowledgements}
\end{center}
We are grateful to our ADIA colleagues Majed Alromaithi,  Dushyant Sharma, Adil Reghai, Oleksiy Kondratiev, Arthur Maghakian, Enzo Busetti, Thomas Schmelzer, Sebastien Lefort, Patrik Karlsson, and other Q-team colleagues and we wish to thank them for their comments and encouragement. 
\vspace{.5 cm}
\begin{center}
\huge{\sc{Appendix}}
\end{center}

\appendix
\section{{\sc Some Useful Lemmas}} \label{sec:Lemma}
We will use the following lemmas in various parts of the papers. These are well known results and proofs can be found in standard linear algebra text books.

\begin{lemma}$\forall ~Z \in {\sym} (\mathbb{R})$,
\be
\label{eq:matrixDer}
{d \over dZ} \tr\( K_1 Z^2 K_2\) = K_1^T K_2^T Z + Z K_1^T K_2^T
\ee
\end{lemma}
\begin{lemma}\label{lem:lemII} $\forall Z_1, Z_2 \in {\sym}(\mathbb{R})$,
\be
Z_1^{-\half}\( Z_1^\half Z_2 Z_1^\half\)^\half Z_1^{-\half} = Z_1^{-1} \( Z_1 Z_2 \)^{\half} = \( Z_2 Z_1 \)^\half Z_1^{-1}
\ee
\end{lemma}
\begin{lemma} \label{lem:Lyapunov} (a) Solution of a special case of Lyapunov equation: If $A$ is a symmetric invertible matrix then  
$Z ={\alpha \over 2} \CA^{-1}$ is the unique solution  $\(\text{for } Z \in \text{Sym}_N(\mathbb{R}) \)$ of the following equation:
\be
\label{eq:LyapunovEq}
\CA\, Z + Z\, \CA = \alpha \mathbb{I} 
\ee
That is, $\forall \CA \in \text{Sym}_N(\mathbb{R})$
\be
\label{eq:Lyapunov}
\CA\, Z + Z\, \CA = \alpha \mathbb{I}_N  \implies Z ={\alpha \over 2} \CA^{-1}
\ee
(b) If $A$ is a symmetric matrix, then
\be
\label{eq:Lyapunovb}
\CA\, Z + Z\, \CA = \alpha \mathbb{I}_N \implies \CA\, Z ={\alpha \over 2} \eye_N = Z\, \CA
\ee

\end{lemma}
\begin{lemma} $\forall Z \in \mathbb{R}^{N \times M}$
\be  
\tr\[\(  Z\, Z^T\)^\half\] = \tr\[\(  Z^T\, Z\)^\half\]
\ee
\end{lemma}

\section{{\sc Lagrangian Form of the Constrained Optimization Problem}} \label{app:LagForm}
This appendix can be skipped by readers who are familiar with Slater conditions and its connection to existence of Lagrange multipliers in a convex optimization problem.
\begin{prop} \label{prop:LagForm}
Let us consider the following optimization problems (with $ \vecX \in \mathbb{R}^d$):
\begin{itemize}
\item Constrained optimization problem:
\bea
\mathop{\min}_{\vecX } \Psi(\vecX), & \\
\nonumber
{\text{\normalfont subject to,}}\qquad \qquad & \\
\label{eq:ConstSetLag}
\vartheta(\vecX) \le \vartheta_0 
\eea
\item Lagrangian form (with $\lambda \ge 0$):
\be
\mathop{\min}_{\vecX} \[\Psi(\vecX) + \lambda \vartheta(\vecX) \]\qquad
\label{eq:LagFormApp}
\ee
\end{itemize}
where $\Psi$ and $\vartheta$ are convex. The above two formulations are equivalent if $\vartheta_0 = 0$ is the {\it only} value for which the constrain set (\ref{eq:ConstSetLag}) is feasible but not strictly feasible.
\end{prop}
\begin{proof}
The statement and a sketch of the proof can be found in \cite{Tibshirani}.  An equivalent form of the proposition can also be found in \cite{Kloft} (in particular, Proposition 12 of \cite{Kloft}). The following proof is a very minor modification of the proof in \cite{Kloft} and this proposition itself is reasonably well known in the the literature on convex optimization. We present the proof here for the convenience of the readers not familiar with this topic.

\noindent
(i) Let $\vecX^{\star}$ be the optimal solution of the constrained optimization problem with $\vartheta_0 \neq 0$. Since $\vartheta_0 = 0$ is the only value for 
which the constrain set (\ref{eq:ConstSetLag}) is feasible but not strictly feasible, the constraint set (\ref{eq:ConstSetLag}) is strictly feasible for $\vartheta_0 \neq 0$. This implies that the
 {\it Slater conditions} are satisfied and strong-duality holds good \cite{Bertsekas, BoydII} ($\Psi$ and $\vartheta$ are convex functions). In this case, we have
 \be
 \vecX^{\star} = \mathop{\text{argmin}}_{\vecX} \(  \Psi(\vecX) + \gamma^{\star} (\vartheta(\vecX) - \vartheta_0) \)
 \label{eq:StrongMaxMin}
 \ee
 where $\gamma^{\star}$ is obtained as follows:
  \be
 \gamma^{\star} = \mathop{\text{argmax}}_{\gamma} \[\mathop{\min}_{\vecX} \(  \Psi(\vecX) + \gamma (\vartheta(\vecX) - \vartheta_0) \)\]
 \label{eq:optGam}
 \ee
Recall that the duality gap is zero when strong-duality holds good and hence solving (\ref{eq:StrongMaxMin}-\ref{eq:optGam}) is equivalent to solving the original constrained optimization problem. Further, 
 $\gamma^{\star} \vartheta_0$ is just a constant term while solving for $\vecX^{\star}$ in (\ref{eq:StrongMaxMin}) and can be dropped. The problem in (\ref{eq:StrongMaxMin}) is equivalent to solving (\ref{eq:LagFormApp})
 with $\lambda = \gamma^{\star}$.  When $\vartheta_0 =0$, the primal problem by itself is equivalent to (\ref{eq:LagFormApp}).  Hence, we have shown that if  $\vecX^{\star}$ is an optimal solution of (\ref{eq:ConstSetLag})
 for some $\vartheta_0 $, there exists a $\lambda \ge 0$ for which it is optimal in (\ref{eq:LagFormApp}). \\
 \noindent
 (ii) We will now show that the reverse statement is also true {\it i.e.},  if  $\vecX^{\star}$ is an optimal solution of (\ref{eq:LagFormApp})
 for some $\lambda \ge 0$, there exists a $\vartheta_0  $ for which it is optimal in  (\ref{eq:ConstSetLag}). By choosing $\vartheta_0 = \vartheta(\vecX^{\star})$ we will have $\vecX^{\star}$ to be optimal in (\ref{eq:ConstSetLag}) as well.  
 This completes the proof of this proposition.
\end{proof}

\section{{\sc Push-forward of a measure}} \label{app:push}

Readers familiar with the notion of {\it push-forward} of a measure can skip this section of the paper.

This appendix provides a formal definition for the push-forward of a measure along a measurable map and a proposition that provides a method for computing the push-forward of a measure. 
We compute the distribution associated with the push-forward of a Gaussian measure along a linear map as an example application of the proposition. There are much simpler techniques to compute this push-forward (as discussed in the main text), but the method described below can be generalized to arbitrary maps and distributions and hence presented here. The advertised definition and proposition are presented below.
\begin{definition}
\label{def:Push}
Given a measure space $(X_1, \sigmaF_1, \meas_1)$, a measurable space $(X_2, \sigmaF_2)$ and a measurable map $\measF: X_1 \mapsto X_2$, the push-forward of $\meas_1$, $\pushF \meas_1$ is defined to be a measure on $\sigmaF_2$.
\end{definition}
\vspace{-.35cm}
\noindent
The following proposition provides a practical definition of push-forward which is useful for computations:
\begin{prop}\label{prop:PushF}
Let $(X_1, \sigmaF_1, \meas_1)$ be a measure space, $(X_2, \sigmaF_2)$ a measurable space, $\measF: X_1 \mapsto X_2$ a measurable map and a $\sigmaF_2$-measurable and integrable function on $X_2$, then  $\pushF \meas_1$ satisfies the following:
\be
\int_{X_2} g(x_2) d(\pushF \meas_1) = \int_{X_1} g(\measF(x_1)) d \meas_1
\label{eq:PropBog}
\ee
\end{prop}
\begin{proof}
The proof of this proposition can be found in \cite{Bogachev} and the proposition is sometimes known as the change-of-variables theorem \cite{McCann}.
\end{proof}
\vspace{-.35cm}
As an example application, we will compute the distribution associated with the push-forward of a Gaussian measure along a linear map $\viewP$, using Proposition (\ref{prop:PushF}). 
The precise calculations are described below:

Let $\vecX \sim \CCN\(\vmu, \cov \)$, $f_\CiX$ be the multi-variate normal distribution associated with $\vecX$ and $\viewP$ be the linear map $\viewP: \vecX \mapsto \vecY$ such that $\vecY = \viewP \vecX$.{\footnote{In equation (\ref{eq:PropBog}), we set $x_1$ to $\vecX$, $x_2$ to $\vecY$, $\measF$ to $\viewP$. The probability density associated with the measure $\meas_1$ is $f_\CiX$.}} The push-forward  $\pushP[f_{\CiX}](\vecY) $ is computed by choosing $g(\vecY')$ as the indicator function $\mathbf{1}_{\vecY' < \vecY}$ and differentiating both sides of equation (\ref{eq:PropBog}) with respect to $\vecY$ as shown below:
\be
\pushP[f_{\CiX}](\vecY) = \int {d^N \vecX \over \sqrt{(2\pi)^N \det \cov} }\,  \(e^{-{\half} (\vecX -\vmu )^T \cov^{-1} (\vecX -\vmu ) }\, \delta\(  \viewP \vecX-\vecY\, \)\) 
\label{eq:PushDelta}
\ee
If $\viewP$ is a invertible matrix (hence a square matrix), the above integrand is straightforward and evaluates to the Gaussian distribution associated with $\CCN\(\viewP \vmu,\,  \viewP \cov \viewP^T \)$.  We will now show that the distribution associated with the push-forward measure has the same form even when $\viewP$ is not a square matrix. This can be done by introducing the Fourier representation of the Dirac delta in equation (\ref{eq:PushDelta}) as shown below:
\be
\pushP[f_{\CiX}](\vecY) = \int {d^N\vec{\mcv}\over (2 \pi)^N} \,  e^{-\iota \vec{\mcv}^T (\vecY - \viewP \vecX) } \int d^N \vecX {\(e^{-\half (\vecX -\vmu )^T \cov^{-1} (\vecX -\vmu ) }\)\over \sqrt{(2\pi)^N \det \cov} } 
\ee
Note that the integrand is still quadratic in $\vecX$. Hence the integral over $\vecX$ is a simple Gaussian integral and can be evaluated by reorganizing the integrand as shown below:
\be
\pushP[f_{\CiX}](\vecY) = \int {d^N\vec{\mcv}\over (2 \pi)^N} \,  e^{-\iota \vec{\mcv}^T (\vecY - \viewP \vmu) -\half \vec{\mcv}^T \viewP \cov \viewP^T \vec{\mcv}  } \int d^N \vecX {\(e^{-\half \(\vecX-(\vmu + \iota \cov \viewP^T 
\vec{\mcv})\)^T \cov^{-1} \(\vecX -(\vmu + \iota \cov \viewP^T \vec{\mcv})\)  }\)\over \sqrt{(2\pi)^N \det \cov} }
\ee
The integral over $\vecX$ is a straightforward Gaussian integral and we obtain the following simplified expression for the push-forward distribution. 
\be
\pushP[f_{\CiX}](\vecY) = \int {d^N\vec{\mcv}\over (2 \pi)^N} \,  e^{-\iota \vec{\mcv}^T (\vecY - \viewP \vmu) -\half \vec{\mcv}^T \viewP \cov \viewP^T \vec{\mcv}  }
\label{eq:pushChar}
\ee
The expression in equation (\ref{eq:pushChar}) is the inverse Fourier transform of the characteristic function of $\CCN\(\viewP \vmu, \right.$ $\left. \viewP \cov \viewP^T \)$. Note that the above derivation is applicable even when $ \viewP \cov \viewP^T$ is singular or when $\viewP$ is degenerate. 

\section{{\sc Wasserstein Distance Between Two Gaussian Measures}} \label{app:WassDist}

The following details appear in Lemma 2 of \cite{Givens} and we present it here again for the sake of clarity and also to emphasize that the Wasserstein distance is well defined even if the views matrix $\viewP$ is degenerate. 

The $L_2$ Wasserstein Distance $\scW_{2}$ between two distributions $g_1$ and $g_2$, is defined as follows:
\be
\scW_2^2(g_1,g_2) =  \mathop{\min}_{\gamma\in\mathscr{G}(g_1,g_2)} \mathbb{E}_{\vecX,\vecY\sim \gamma}\left[\| \vecX - \vecY\|^2\right]
\label{eq:WDDef}
\ee
where $\mathscr{G}(g_1,g_2)$ denotes the set of all joint probability distributions whose marginals are $g_1$ and $g_2$.
In the following, we will assume that $g_1$ and $g_2$ are Gaussian distributions unless otherwise specified. We will also assume that $g_2$ is  non-degenerate while $g_1$ is allowed to be degenerate.{\footnote{ The proof can be modified to allow both $g_1$ and $g_2$ to be degenerate.}} 

 Now, let us rewrite equation (\ref{eq:WDDef}) as follows:
\bea
\nonumber
\mathbb{E}_{\gamma}\left[\| \vecX - \vecY\|^2\right] = \|\vmu_1 - \vmu_2\|^2_2 & +&  \mathbb{E}_\gamma\left[ (\vecX - \vmu_1)^T (\vecX - \vmu_1)  +  (\vecY - \vmu_2)^T (\vecY - \vmu_2) \right. \\
& &  \left.   - 2 (\vecY -\vmu_2)^T(\vecX - \vmu_1)\right]
\eea
where $\vec{\mu}_i$ denotes the mean of the Gaussian distribution $g_i$. Further,  
$$\mathbb{E}_{\gamma}\left[\| \vecX - \vecY\|^2\right] = \|\vmu_1 - \vmu_2\|^2_2 + \tr(\cov_1 + \cov_2 - 2K)$$
where $\cov_i$  is the covariance matrix associated with the Gaussian $g_i$ and $K = \mathbb{E}_{\gamma}[ (\vecY -\vmu_2)^T(\vecX - \vmu_1)] $.  From the assumption that $g_2$ is non-degenerate it follows that $\cov_2 \in \text{Sym}^{++}_{N}(\mathbb{R})$ and is invertible. The covariance matrix can be obtained by evaluating the Hessian of the characteristic function which well defined even when the views matrix is degenerate. 

Let us introduce the matrix $\mathfrak{C}$ which is defined as follows:
\be
\mathfrak{C} = \[ \begin{matrix}  \mathbb{E}_{\gamma}[ (\vecX -\vmu_1)(\vecX - \vmu_1)^T]  &  \mathbb{E}_{\gamma}[ (\vecX -\vmu_1)(\vecY - \vmu_2)^T] \\
 \mathbb{E}_{\gamma}[ (\vecY -\vmu_2)^T(\vecX - \vmu_1)] &  \mathbb{E}_{\gamma}[ (\vecY -\vmu_2)(\vecY - \vmu_2)^T]\\
\end{matrix} 
\] =  \[ \begin{matrix}  \cov_1 &K^T \\
K&  \cov_2\\
\end{matrix} 
\]
\ee
This matrix $\mathfrak{C}$ is clearly positive definite and hence the Schur complement $\mathfrak{C}/ \cov_2$ is positive semi-definite. That is,
\be
\cov_1 - K^T\cov_2^{-1}K \succeq 0
\ee
Note that the matrix $\cov_2$ needs to be invertible so that $\mathfrak{C}/ \cov_2$ is well-defined. This follows from our assumption that $g_2$ is non-degenerate. Let us denote $\cov_1 - K^T\cov_2^{-1}K $ by $\CS$. Then we have 
\be 
K^T\cov_2^{-1}K = \cov_1 - \CS
\label{eq:KRicatti}
\ee
 Let us denote the diagonalization of  $ \cov_1 - \CS$ as follows:
\be
[ \cov_1 - \CS]_{ij} = \sum_{ik}\sum_{kj} U_{ik} \Lambda^2_{kk} U_{jk}~ \text{that is, } \cov_1 - \CS = U \Lambda^2 U^T
\label{eq:DiagSchur}
\ee
where $\Lambda^2$ denotes the diagonal matrix of eigenvalues and $U$ denotes the matrix of the corresponding eigenvectors. If $\text{rank}( \cov_1 - \CS ) = r < N$ , then $\Lambda^2=\text{diag}(\lambda_1^2, \lambda_2^2, \dots \lambda_r^2)\oplus 0_{N - r}$ and $U = [U_r, U_{N-r}]$. Equation (\ref{eq:DiagSchur}) can now be written as follows:
\be
 \cov_1 - \CS = U_r \Lambda_r^2 U_r^T
\label{eq:decSchur}
\ee
where $\Lambda_r^2 = \text{diag}(\lambda_1^2, \lambda_2^2, \dots \lambda_r^2)$. 
Using equation (\ref{eq:KRicatti}) and (\ref{eq:decSchur}), we get
\be
K^T \cov_2^{-1} K = U_r \Lambda_r^2 U_r^T \implies \( \cov_2^{-\half} K \Lambda_r^{-1} U_r  \)^T  \( \cov_2^{-\half} K \Lambda_r^{-1} U_r  \) = \mathbb{I}_{r}
 \implies K = \cov_2^{\half} \CO_r \Lambda_r U_r^T
\ee
for some $\CO_r$ is an $N \times r$ matrix such that $\CO^T_r \CO_r = \mathbb{I}_r$. Note that this is an orthogonality condition on $\CO_r$ in $N$ dimensions. We can lift $\CO_r$ to an $N-$dimensional orthogonal matrix $\CO$ and obtain the following condition:
\be
K = \cov_2^{\half} \CO \Lambda U^T
\ee
We have used the fact that $\Lambda = \Lambda_r \oplus 0_{N-r}$ to obtain the above equation. Now we can work with $\CO$ which is an $N \times N$ matrix such that $\CO^T \CO = \mathbb{I}_N$. 

To find the minimum value of the objective defined in equation (\ref{eq:WDDef}), we need to minimize $-2\Tr(K)$ subject to the condition $\CO \CO^T = \mathbb{I}$.  We will introduce a matrix Lagrange multiplier ${\mathcal H}$ to enforce the orthogonality condition on the matrix $\CO$. The modified objective function with the Lagrange multiplier is given by
\be
{\cal L} = - 2 \Tr[  \CO^T  \cov_2^{\half}  U \Lambda ] + \Tr[{\mathcal H} .(\CO^T \CO - \mathbb{I})]
\ee
After solving for $\CO$ and the Lagrange multiplier $\mathcal{H}$ we get,
\be
\CO = {\mathcal{H}}^{-1}  \cov_2^{\half}  U \Lambda,  \quad {\mathcal H} = \(\( \cov_2^{\half}  U \Lambda\) \(\cov_2^{\half}  U \Lambda\)^T\)^{\half}
\ee
Substituting for $K$ in the definition of Wasserstein distance we get,
\be
\scW_2^2(g_1,g_2) =  \mathop{\min}_{\CS} \left[ \|\vmu_1 - \vmu_2\|^2_2 + \tr\(\cov_1 + \cov_2 - 2\(\cov_2^{\half} \( \cov_1 - \CS\) \cov_2^{\half}\)^{\half}\)\right]
\ee
The minimum value is achieved when $\CS=0$ since $\CS$ is a positive definite matrix. Therefore,
\be
\scW_2^2(g_1,g_2) =  \|\vmu_1 - \vmu_2\|^2_2 + \tr\(\cov_1 + \cov_2 - 2\(\cov_1^{\half} \cov_2 \cov_1^{\half}\)^{\half}\)
\label{eq:appWassDist}
\ee
We would like to emphasize that the above derivation is valid even when the distribution $g_1$ is degenerate. Hence, the above derivation is applicable even when the push-forward of the prior distribution is degenerate. It is possible to modify the above derivation to show that $\scW_2^2(g_1,g_2)$ when both $g_1$ and $g_2$ are degenerate, by changing $\cov_2 \rightarrow \cov_2 + \delta^2 \mathbb{I}$ and finally take the limit as $\delta \rightarrow 0$. 
In the extreme case when $\cov_1 = \cov_2 = 0$, $\scW_2^2(g_1,g_2)$ is equivalent to the distance between two point masses. However, for the analysis in the rest of the paper we will assume any form of degeneracy arises only from the degeneracy of covariance matrices of the form $\viewP \cov \viewP^T$ for some $\cov \in \text{Sym}^{++}_\Nprior(\mathbb{R})$. This could arise from $\viewP$ having identical rows for example.

\section{{\sc Proof of Theorem \ref{thm:Main}: Main Result}}\label{app:OptUpdate}
In this section, we present the proof of theorem \ref{thm:Main} which is the main result of this paper.
\subsection{Details of Computing Optimal Updates}
\begin{proof}
The Wasserstein distance between Gaussian distributions can be written as sum of Euclidean distance between the drifts and the Bures distance between covariance matrices. Hence the cost function in equation (\ref{eq:WDDist}) can be written as:
\be
{\CL}_{GWB} = {\CL}_{\dr} \left[  \vmm_\update; \vmu_\prior, \vnu_{\view}, \viewP \right] + {\CL}_{\covar}\[ \cov_\update; \cov_\prior,  \cov_\view, \viewP  \]
\ee
where
\be
 {\CL}_{\dr} \left[  \vmm_\update; \vmu_\prior, \vnu_{\view}, \viewP \right] = \| \vmm_\update - \vmu_\prior \|^2 + \lambda \| \viewP \vmm_\update - \vnu_{\view}\|^2 
\ee
and
\begin{multline}
\label{eq:LCovar}
\qquad \qquad{\CL}_{\covar}\[ \cov_\update; \cov_\prior,  \cov_\view, \viewP  \]= \tr\( \cov_\prior + \cov_\update - 2 \(\cov^\half_\prior \cov_\update \cov^\half_\prior \)^\half\) +\\
 \lambda \tr\( \cov_\view +\viewP \cov_\update\viewP^T - 2 \(\cov^\half_\view \viewP \cov_\update \viewP^T \cov^\half_\view \)^\half\) \qquad \qquad
\end{multline}
Minimizing ${\CL}_{GWB} $ with respect to $\vmm_\update$ and $\cov_\update$ boils down to minimizing ${\CL}_\dr$ with respect to $\vmm_\update$ and $\CL_\covar$ with respect to $\cov_\update$.  Minimizing ${\CL}_\dr$  with respect to $\vmm_\update$ is rather straightforward and yields the following equation:
\be
(\vmm_\update - \vmu_\prior) + \lambda \viewP^T(\viewP \vmm_\update - \vnu_{\view}) = 0
\ee
Simple algebraic manipulation of the above equation yields the expression in equation (\ref{eq:optDrift}). Minimization of ${\CL}_{\covar}$ is slightly more involved and rest of this appendix is dedicated to finding the optimal $\cov_\update$. 
To minimize ${\CL}_{\covar}$ , it seems convenient to employ the following change of variables:
\bea
\label{eq:covToX}
X &=& \(\cov^{\half}_\prior \cov_\update \cov^{\half}_\prior\)^\half \implies \cov_\update = \cov^{-\half}_{\prior} X^2\, \cov^{-\half}_{\prior} \\
Y &=& \(\cov^{\half}_\view \viewP \cov_\update \viewP^T \cov^{\half}_\view\)^\half \implies \viewP \cov_\update \viewP^T = \cov^{-\half}_{\view} Y^2\, \cov^{-\half}_{\view} 
\eea
where $X,~Y \in \text{Sym}(\mathbb{R})$. After the change of variables, the minimization of $\CL_\covar$ can then be recast into the constrained optimization problem:
\bea
\nonumber
\qquad (X_\optU,Y_\optU)&=&\mathop{\text{argmin}}_{X,Y \in {\text{Sym}(\mathbb{R)}}} \Bigg[\tr\( \cov_\prior + \cov^{-\half}_\prior X^2 \, \cov^{-\half}_\prior - 2X\)
+\\
& &\qquad \qquad \lambda \tr\( \cov_\view + \cov^{-\half}_\view Y^2 \, \cov^{-\half}_\view - 2Y\) \Bigg] \\
\nonumber
{\text{subject to}}  \qquad \qquad & & \\
\label{const:XY}
\qquad \cov^{-\half}_\view Y^2 \cov^{-\half}_\view &=& \viewP \cov^{-\half}_\prior X^2 \cov^{-\half}_\prior \viewP^T
\eea
To solve the constrained minimization problem, we introduce a matrix Lagrange multiplier $\CM$ and minimize the following modified cost function:
\bea 
\nonumber
{{\CL}}\[X,Y, \CM\]  & = & \Bigg[\tr\( \cov_\prior + \cov^{-\half}_\prior X^2 \, \cov^{-\half}_\prior - 2X\)  +
  \lambda \tr\( \cov_\view + \cov^{-\half}_\view Y^2 \, \cov^{-\half}_\view - 2Y\) \Bigg] \\
  \label{eq:XYLagrangian}
& + & \tr\[ \CM.\( \cov^{-\half}_\view Y^2 \cov^{-\half}_\view - \viewP \cov^{-\half}_\prior X^2 \cov^{-\half}_\prior \viewP^T\) \]
\eea
The Lagrange multiplier matrix $\CM$ is symmetric since $\cov^{-\half}_\view Y^2 \cov^{-\half}_\view - \viewP \cov^{-\half}_\prior X^2 \cov^{-\half}_\prior \viewP^T$ is symmetric.
 The  modified cost function in (\ref{eq:XYLagrangian}) is minimized by setting the gradients with respect to $X$ and $Y$ to zero and the Lagrange multiplier $\CM$ is obtained by enforcing the constraint in (\ref{const:XY}).
The gradient of ${{\CL}}\[X,Y, \CM\] $ with respect to $X$ and $Y$ is computed using the identity in equation (\ref{eq:matrixDer}). Setting these gradients to zero we get,
\bea
\label{eq:XcovP}
X \(\cov^{-1}_\prior - \cov^{-\half}_\prior \viewP^T \CM \viewP \cov^{-\half}_\prior\) +  \(\cov^{-1}_\prior - \cov^{-\half}_\prior \viewP^T \CM \viewP \cov^{-\half}_\prior\) X &=& 2 \mathbb{I}_{\Nprior} \\
\label{eq:YcovV}
Y \(\lambda \cov^{-1}_\view + \cov^{-\half}_\view  \CM \cov^{-\half}_\view\) + \(\lambda \cov^{-1}_\view + \cov^{-\half}_\view  \CM \cov^{-\half}_\view\) Y \quad ~ &=& 2 \lambda \mathbb{I}_{\Nview}
\eea
The above equations are special cases of Lyapunov equation. If $(\eye_\Nprior - \viewP^T \CM \viewP)$ is invertible then the solution of (\ref{eq:XcovP}) can be written as shown below:
\be
\label{eq:Xsol}
X =  \(\cov^{-1}_\prior - \cov^{-\half}_\prior \viewP^T \CM \viewP \cov^{-\half}_\prior\)^{-1} = \cov^{\half}_\prior(\eye_\Nprior - \viewP^T \CM \viewP)^{-1} \cov^{\half}_\prior
\ee
We have used equation (\ref{eq:Lyapunov}) or Lemma \ref{lem:Lyapunov}(a) to obtain the above solution. 
It must be clear from the definition of $X$ that $(\eye_\Nprior - \viewP^T \CM \viewP)$ is invertible iff $\cov_\prior$ and $\cov_\update$ are invertible. Hence, the validity of this assumption can be verified only after solving for $\cov_\update$. We will show at the end of this section that $\cov_\update$ is indeed invertible and hence invertibility of $(\eye_\Nprior - \viewP^T \CM \viewP)$ is justified. If $(\lambda \eye_\Nview  + \CM )$ is invertible, then the Lyapunov equation (\ref{eq:YcovV}) yields $Y = \lambda\cov^{\half}_\view(\lambda \eye_\Nview  + \CM )^{-1} \cov^{\half}_\view$ as the unique solution, however invertibility of $(\lambda \eye_\Nview  + \CM )$ is not justified if the views matrix $\viewP$ is degenerate. Fortunately, we can derive the optimal update for covariance $\cov_\update$ without inverting $(\lambda \eye_\Nview  + \CM )$. Using Lemma \ref{lem:Lyapunov}(b) we get,
\be
\label{eq:Ysol}
\cov^{-\half}_\view \(\lambda \eye_\Nview +   \CM \) \cov^{-\half}_\view Y =  \lambda \eye_\Nview = Y \cov^{-\half}_\view \(\lambda \eye_\Nview +   \CM \) \cov^{-\half}_\view 
\ee
Now, to determine the optimal values of $X$ and $Y$ we need to determine $\CM$ in equation (\ref{eq:Xsol}) and determine $Y$ using the constraint in equation (\ref{const:XY}). In fact, $\CM$, or rather $\viewP^T \CM \viewP$ is determined by enforcing the constraint in equation (\ref{const:XY}). Using equation (\ref{eq:Ysol}) we get,
\be
\label{eq:PrePTMP}
(\lambda \eye_\Nview  + \CM ) \cov^{-\half}_\view Y^2  \cov^{-\half}_\view (\lambda \eye_\Nview  + \CM ) = \lambda^2  \cov_\view 
\ee
By making use of the constraint in (\ref{const:XY}), we first eliminate $\cov^{-\half}_\view Y^2  \cov^{-\half}_\view$ and we then make use of equation (\ref{eq:Xsol}) in the resulting expression to get,
\be
\label{eq:XYConstPTMP}
(\lambda \eye_\Nview  + \CM ) \viewP (\eye_\Nprior - \viewP^T \CM \viewP)^{-1} \cov_\prior (\eye_\Nprior - \viewP^T \CM \viewP)^{-1} \viewP^T  (\lambda \eye_\Nview  + \CM ) = \lambda^2  \cov_\view 
\ee
\noindent
It is convenient to introduce a matrix ${\cal U}$ such that,
\be
(\lambda \eye_\Nview  + \CM ) \viewP (\eye_\Nprior - \viewP^T \CM \viewP)^{-1}  = \lambda \cov^{\half}_{\view} {\cal U} \cov^{-\half}_\prior
\label{eq:Udef}
\ee
The precise properties of the matrix ${\cal U}$ is not very important here as this will be eliminated in the following steps.
Multiplying both sides of equation (\ref{eq:Udef}) by $\viewP^T$, we get the following result:
\be
\label{eq:GPrel}
\(\mathbb{I}_\Nprior + G \) \viewP^T {\cal M} \viewP = \(G - \lambda \viewP^T \viewP\)  \quad \implies \viewP^T {\cal M} \viewP = \(\eye_\Nprior + G\)^{-1} \(G - \lambda \viewP^T \viewP \), 
\ee
where, $G$ is defined as follows
\be
\label{eq:Gdef}
 G =  \lambda \viewP^T \cov^{\half}_{\view} {\cal U} \cov^{-\half}_\prior
 \ee
Equation (\ref{eq:GPrel}) provides an expression for $\viewP^T {\cal M} \viewP$ in terms of the matrix $G$. The matrix $G$ contains an unknown unitary matrix and we will now find an alternate expression for $G$ by using equation (\ref{eq:GPrel}) and the constraint ${\cal M} = {\cal M}^T$. Using the fact that $\viewP^T \CM \viewP = (\viewP^T \CM \viewP)^T $ whenever $\CM = \CM^T$ in equation (\ref{eq:GPrel}) we get,
\be
\(\mathbb{I}_{\Nprior} + G \) \( G^T - \lambda \viewP^T \viewP\) = \(G - \lambda \viewP^T \viewP\) \( \eye_\Nprior + G^T \)
\ee
From the above condition we can conclude that $G$ can be written as
\be
\label{eq:defS}
G = S. W, \qquad {\text {where}}~ S = S^T,\quad W =  \(\WW\)^{-1}
\ee
We can compute $X$ from (\ref{eq:Xsol}) if $(\eye_\Nprior - \viewP^T \CM \viewP)$ is known. Using equations (\ref{eq:GPrel}) and (\ref{eq:defS}) we get,
\be
\label{eq:defSII}
\eye_\Nprior - \viewP^T \CM \viewP = (\eye_\Nprior + G)^{-1} \(\eye_\Nprior + \lambda \viewP^T \viewP\) = \( \eye + S.W\)^{-1} W^{-1}
\ee
We will now describe the procedure to determine $S$ in the above equation. By using the definition of $G$ in equation (\ref{eq:Gdef}), equations (\ref{eq:Udef}) and (\ref{eq:XYConstPTMP}) we get,
\be
\label{eq:Gconst}
G \cov_\prior G^T = \lambda^2 \viewP^T \cov_\view \viewP
\ee
After defining $A = W \cov_\prior W$, equation (\ref{eq:Gconst}) can now be written as follows:
\be
S.A.S = \lambda^2 \viewP^T \cov_\view \viewP  \implies (A^{\half} S A^{\half} )(A^{\half} S A^{\half}) =  \lambda^2 A^{\half} \viewP^T \cov_\view \viewP A^{\half}
\ee
Hence $S$ is given by,
\be
\label{eq:Ssol}
S = \mathcal{s} \lambda A^{-\half} \(A^{\half} \viewP^T \cov_\view \viewP A^{\half}\)^{\half} A^{-\half}
\ee
where $\mathcal{s} = \pm 1$. 
Using equations (\ref{eq:Ssol}), (\ref{eq:defSII}), (\ref{eq:Xsol}) and (\ref{eq:covToX}) we get,
\be
\label{eq:covS}
\cov({\cal s}) =  \(\IPPT+ \BS\)\cov_\prior\(\IPPT+ \BS\)
\ee
where, $\BS= \BS^T $ and it is given by
\be
\label{eq:BisWSW}
\BS = \mathcal{s} \lambda \IPPT A^{-\half} \(A^{\half} \viewP^T \cov_\view \viewP A^{\half}\)^{\half} A^{-\half}  W = W. S. W
\ee
In Appendix \ref{app:PropMin}, we show that ${\cal s}=1$ for ${\cal L}_{\covar}$ to be a minimum at $\cov_\update = \cov({\mathcal{s}}) $. 
Hence, 
\be
\cov_{\optU} =  \(\IPPT+ \EuScript{B}\)\cov_\prior\(\IPPT+ \EuScript{B}\)
\ee
where, $\EuScript{B} = \EuScript{B}^T $ and it is given by
\be
\EuScript{B} = \lambda \IPPT A^{-\half} \(A^{\half} \viewP^T \cov_\view \viewP A^{\half}\)^{\half} A^{-\half}  W , \qquad
A = W \cov_\prior W
\ee
This completes the proof the Theorem (\ref{thm:Main}).
\end{proof}
\subsection{{\sc Proof of } ${\cal s}=1$}\label{app:PropMin}

We present a heuristic argument for the proof first to provide an intuition behind the proof which requires tedious algebra.
To fix ${\cal s}$ we evaluate ${\cal L}_{\covar}$ at $\cov_\update = \cov({\mathcal{s}}) $  and minimize with respect to ${\cal s}$. We will show that  ${\cal s} $ should be $\mathcal{s} = 1$ for ${\cal L}_{\covar}$ to be minimized. For the purpose of this heuristic argument we will assume $\viewP = \eye_{\Nprior}$, $\cov_{\prior,\view} = \text{\sc{Diag}}(\sigma^2_{\prior,\view})$. In this case ${\CL}_{\covar}$ can be written as:
\be
\CL_{\cal s} = \(\sigma_\prior - {\sigma_\prior + \lambda {\cal s} \sigma_\view \over 1 + \lambda} \)^2 + \lambda \(\sigma_\view - {\sigma_\prior + \lambda {\cal s} \sigma_\view \over 1 + \lambda} \)^2
\ee
After a little bit of algebra, we can infer that $\CL_{\cal s} $ is minimized when ${\cal s}=1$. The same conclusion can be reached for a general $\cov_{\prior, \view}$ and $\viewP$, but the algebra is more tedious and we present the proof for a general $\cov_{\prior, \view}$ and $\viewP$ below.

\begin{proof} We will prove this result for the case when the views matrix is not degenerate and $\viewP^T \cov_\view \viewP$ is not degenerate. The proof for a general views matrix can be modified by introducing a regulating parameter $\delta$ and then taking the limit $\delta \rightarrow 0$. 
Or alternatively, the proof can be modified by introducing the Moore-Penrose inverse wherever necessary. 

We will first prove that $W + \BS$ is positive definite.
We know from that $X$ is symmetric and positive definite by definition. Hence $X$ can be written as $\Upsilon\,\Upsilon^T$ for some $\Upsilon$. 
Then, it follows from equation (\ref{eq:Xsol}) that $\(\eye_\Nprior - \viewP^T \CM \viewP\)^{-1}$ is also positive definite.
Using equation (\ref{eq:defSII}) we get,
\be
\label{eq:posWBS}
\(\eye_\Nprior - \viewP^T \CM \viewP\)^{-1} = W + \BS \quad  \implies W + \BS \succ 0
\ee
Let $Q = \viewP^T \cov_\view \viewP$ for convenience. Using Lemma (\ref{lem:lemII}) repeatedly, $S$ can be written as shown below
\be
\label{eq:StoGamma}
S = \mathcal{s} \lambda A^{-\half} \(A^{\half} Q A^{\half}\)^{\half} A^{-\half} 
= \mathcal{s} \lambda W^{-1} \cov^{-\half}_\prior\(\cov^{\half}_\prior W.Q.W \cov^{\half}_\prior\)^{\half}\cov^{-\half}_\prior W^{-1} = \mathcal{s} \lambda \Gamma^{-1}
\ee
where,
\be
\label{eq:Gamma}
\Gamma = W \cov^{\half}_\prior\(\cov^{\half}_\prior W.Q.W \cov^{\half}_\prior\)^{-\half}\cov^{\half}_\prior W
\ee
Similarly, we know that $Y$ is positive semi-definite. Then from equation (\ref{eq:Ysol}), we can conclude that $\lambda \eye_\Nview + \CM$ is also positive semi-definite.{\footnote{Equation (\ref{eq:Ysol}) implies 
\be
\vec{z}^T  \(\lambda \eye_\Nview +   \CM \) \cov^{-\half}_\view Y \cov^{-\half}_\view \(\lambda \eye_\Nview +   \CM \)  \vec{z} =  \lambda \vec{z}^T  \(\lambda \eye_\Nview +   \CM \)  \vec{z}, \qquad \text { for any } \vec{z} \in \mathbb{R}^\Nview
\ee
\hspace{.5cm} Note that LHS is greater than or equal to zero because $Y$ is positive semi-definite. Hence $\(\lambda \eye_\Nview +   \CM \) \succcurlyeq 0 $.  }}
Using equation (\ref{eq:GPrel}) we get,
\be
\label{eq:LamIPlusM}
\viewP^T\(\lambda \eye_\Nview + \CM\)\viewP = \(\eye_\Nprior + G\)^{-1} G W^{-1}
\ee
Now, from the definition of $S$ in equation (\ref{eq:defS}), equations (\ref{eq:StoGamma}) and (\ref{eq:LamIPlusM}) we get,
\be
\label{eq:posGamma}
\viewP^T\(\lambda \eye_\Nview + \CM\)\viewP  =  \lambda \(\mathcal{s}\Gamma +  \lambda W\)^{-1} \implies \(\mathcal{s} \Gamma +  \lambda W\) \succcurlyeq 0
\ee
That is, $ \(\mathcal{s} \Gamma +  \lambda W\)$ is positive semi-definite.{\footnote{Note that, the positive semi-definiteness holds good even if the inverse is replaced by Moore-Penrose inverse in the degenerate case. In the non-degenerate case, $ \(\mathcal{s} \Gamma +  \lambda W\) \succ 0$. }} Now, from equations (\ref{eq:covS}), (\ref{eq:StoGamma}) and (\ref{eq:BisWSW}) we get,
\be
\label{eq:covSGamma}
\cov(\mathcal{s})= (W + s \lambda W \Gamma^{-1} W) \cov_\prior  (W + \mathcal{s} \lambda W \Gamma^{-1} W) = (\mathcal{s} \Gamma +  \lambda W) \viewP^T \cov_\view \viewP (\mathcal{s} \Gamma +  \lambda W) 
\ee
We have used the fact that $\mathcal{s}^2 = 1$ to obtain the above equation. From equations (\ref{eq:covS}) and (\ref{eq:covSGamma}) we get,
\be
\label{eq:sqrtI}
\tr\(\(\cov^{\half}_\prior \covs \cov^{\half}_\prior \)^{\half}\) = \tr\( \cov^{\half}_\prior ( W + \BS)  \cov^{\half}_\prior \)
\ee
\be
\label{eq:sqrtII}
\tr\(\(\cov^{\half}_\view \viewP \covs \viewP^T \cov^{\half}_\view \)^{\half}\) = \tr\( \cov^{\half}_\view \viewP (\mathcal{s} \Gamma +  \lambda W)  \viewP^T \cov^{\half}_\view \)
\ee
In the equation (\ref{eq:sqrtI}), positive square root was chosen positive definiteness of $ W + \BS$ proved in equation (\ref{eq:posWBS}) and in equation (\ref{eq:sqrtII}), it was chosen using positive definiteness of $\(\mathcal{s} \Gamma +  \lambda W\)$ proved in equation (\ref{eq:posGamma}). Using equations (\ref{eq:sqrtI}) and (\ref{eq:sqrtII}) to simplify ${\cal L}_{\covar}$ in (\ref{eq:LCovar}) we get,{\footnote{The algebra is slightly tedious, but if we only focus on the $\mathcal{s}$ dependent terms, the task of simplifying becomes less laborious. 

\hspace{.3cm}Lemma \ref{lem:lemII} was used again.}}
\be
\label{eq:SimplLCovar}
{\CL}_{\covar}\[ \covs; \cov_\prior,  \cov_\view, \viewP  \] = {\CL}^{(0)}_{\covar} - \mathcal{s} \tr\( \viewP^T \cov_\view \viewP W\) = {\CL}^{(0)}_{\covar} - \mathcal{s} \tr\( \cov^{\half}_\view \viewP W \viewP^T \cov^{\half}_\view \)
\ee
where ${\CL}^{(0)}_{\covar}$ is a term independent of $\mathcal{s}$ and we have used ${\mathcal{s}}^2=1$ to obtain the above expression. Since $\( \cov^{\half}_\view \viewP W \viewP^T \cov^{\half}_\view \)$ is positive definite, 
$${\CL}_{\covar}\[ \cov(\mathcal{s} = 1); \cov_\prior,  \cov_\view, \viewP  \]< {\CL}_{\covar}\[ \cov(\mathcal{s}=-1); \cov_\prior,  \cov_\view, \viewP  \]$$
Hence ${\CL}_{\covar}\[ \covs; \cov_\prior,  \cov_\view, \viewP  \]$ is minimized at $\mathcal{s}=1$.
\end{proof}
\section{{\sc Allocation Methodologies Summary}} \label{app:AllocMethod}
In the following, we present the details of the four allocation methodologies $\BLI\, \BLII,\, \GWI,$ and $\GWII$:
\subsection{$\text{\sc BL}_{\text{I}}$ Allocation Method} \label{app:BLImethod}
  \begin{center}
  \scalebox{0.7}{
    \begin{minipage}{\linewidth}
\noindent
\rule{1.05\textwidth}{0.5pt}
\noindent
\rule{1.05\textwidth}{0.5pt}
\begin{algorithm}[H]
\begin{center}
\vspace{-.5cm}
\TitleOfAlgo{$\text{\sc BL}_{\text{I}}$ Allocation Method}
\end{center}

\noindent
{\bf Input}: {$\vmu_\drift$, $\hcov_\Rsc$, $\viewP$, $\nu_\view$, $\cov_\view$, $\tau$, $\gammaR$}

\vspace{.1cm}
{\bf Method}:
\begin{itemize}
 \item 
 Using equation (\ref{BLcov}) we compute $ \cov^{(\vRet)}_{\BL}$:
 \be
  \cov^{(\vRet)}_{\BL} =  \( \(\tau \hat{\cov}_\Rsc\)^{-1} +   \viewP^T \cov_\view^{-1} \viewP\)^{-1}
 \ee
 \item {\sc Covariance Update}: Using equation (\ref{BLCUpdate}) we set $\cov^{(\text{\sc BLI})}_\Est$ to $ \hat{\cov}_{\subRet|\view}$: 
 \be
\cov^{(\BLI)}_\Est \leftarrow  \hat{\cov}_{\subRet|\view}  = \hat{\cov}_{\Rsc} +  \cov^{(\vRet)}_{\BL}
\label{eq:BL1covE}
\ee
\item {\sc Drift Update}: From equation (\ref{BLmu}) and \ref{BLCUpdate}):
\be
\vmm^{(\BLI)}_\Est \leftarrow\vmu_{\BL} = \cov^{(\vRet)}_{\BL}\( \(\tau \hat{\cov}_\Rsc\)^{-1} \vmu_\drift+ \viewP^T\cov_\view^{-1} \vnu_{\view} \)
\label{eq:BL1muE}
\ee
\item {\sc Optimal Weights}: We compute optimal weights with the \BLmodelI \, update as follows:
\be
\wBLI = \text{{\sc MVO}}\[\vmm^{(\BLI)}_\Est , \cov^{(\BLI)}_\Est; \gammaR\]
\label{eq:WBLI}
\ee
\end{itemize}
{\bf Result}: Weights $\wBLI$ computed in equation (\ref{eq:WBLI}).
\end{algorithm}
\rule{1.05\textwidth}{0.5pt}
\end{minipage}
}
\end{center}

  \subsection{$\text{\sc BL}_{\text{II}}$ Allocation Method}\label{app:BLIImethod}
    \begin{center}
  \scalebox{0.7}{
    \begin{minipage}{\linewidth}
\noindent
\rule{1.05\textwidth}{0.5pt}
\noindent
\rule{1.05\textwidth}{0.5pt}
\begin{algorithm}[H]
\begin{center}
\vspace{-.5cm}
\TitleOfAlgo{$\text{\sc BL}_{\text{II}}$ Allocation Method}
\end{center}

\noindent
{\bf Input}: {$\hvRet$, $\hcov_\Rsc$, $\viewP$, $\nu_\viewR$, $\cov_\viewR$, $\gammaR$}

\vspace{.1cm}
{\bf Method}:
\begin{itemize}
 \item {\sc Covariance Update}: Using equation (\ref{BLMcov}) we compute $\cov^{(\BLII)}_\Est$ 
 \be
 \label{eq:BLIIcovE}
\cov^{(\BLII)}_\Est \leftarrow \cov^{(\vec{R})}_{\BL'} = \(  \hcov^{-1}_\Rsc+   \viewP^T \cov_\viewR^{-1} \viewP\)^{-1}
\ee
\item {\sc Drift Update}: Corrections to the drift is computed from equation (\ref{BLmu}):
\be
 \label{eq:BLIImuE}
 \vmm^{(\BLII)}_\Est \leftarrow \vmu^{(\vec{R})}_{\BL'} = \( \hat{\cov}^{-1}_\Rsc +   \viewP^T \cov_\viewR^{-1} \viewP\)^{-1}\( \hcov^{-1}_\Rsc \, \hvRet+ \viewP^T\cov_\viewR^{-1} \vnu_{\view} \)
\ee

\item {\sc Optimal Weights}: We compute optimal weights with the \BLmodelII \, update as follows:
\be
\wBLII = \text{{\sc MVO}}\[\vmm^{(\BLII)}_\Est , \cov^{(\BLII)}_\Est; \gammaR\]
\label{eq:WBLII}
\ee
\end{itemize}
{\bf Result}: Weights $\wBLII$ computed in equation (\ref{eq:WBLII}).
\end{algorithm}
\rule{1.05\textwidth}{0.5pt}
\end{minipage}
}
\end{center}

  \subsection{$\text{\sc GWB}_{\text{I}}$ Allocation Method}\label{app:GWBImethod}
  \begin{center}
  \scalebox{0.7}{
    \begin{minipage}{\linewidth}
\noindent
\rule{1.05\textwidth}{0.5pt}
\noindent
\rule{1.05\textwidth}{0.5pt}
\begin{algorithm}[H]
\begin{center}
\vspace{-.5cm}
\TitleOfAlgo{$\text{\sc GWB}_{\text{I}}$ Allocation Method}
\end{center}

\noindent
{\bf Input}: {$\vmu_\drift$, $\hcov_\Rsc$, $\viewP$, $\nu_\viewd$, $\cov_\viewd$, $\tau$, $\gammaR$, $\lambda$}

\vspace{.1cm}
{\bf Method}:
\begin{itemize}
 \item Drift Update:
 \be
 \IPPT = \(\eye_\Nprior + \viewP^T \viewP\)^{-1}
 \ee
\be
 \label{eq:GWImuE}
 \vmm^{(\GBI)}_\Est \leftarrow \vmm_{\GBI} = \IPPT \( \vmu_\drift + {\lambda} \viewP^T  \vnu_{\viewd} \)
\ee
 \item Covariance Update:
 \be
 A_\drift  = \tau \IPPT \hcov_\Rsc \IPPT
 \ee
 \be
 \EuScript{B}_\viewd = \lambda \IPPT A_\drift^{-\half} \(A_\drift^{\half} \viewP^T \cov_\viewd \viewP A_\drift^{\half}\)^{\half} A_\drift^{-\half}  W
 \ee
\be
 \label{eq:GWIcovE}
 \cov^{(\GBI)}_\Est \leftarrow \cov_\GBI = \hcov_\Rsc + \tau\(\IPPT+ \EuScript{B}_{\viewd}\)\hcov_\Rsc\(\IPPT+ \EuScript{B}_{\viewd}\)
  \ee

\item {\sc Optimal Weights}: We compute optimal weights with the \GWBI \, update as follows:
\be
\wGWI = \text{{\sc MVO}}\[\vmm^{(\GBI)}_\Est , \cov^{(\GBI)}_\Est; \gammaR\]
\label{eq:GWBIapp}
\ee
\end{itemize}
{\bf Result}: Weights $\wGWI$ computed in equation (\ref{eq:GWBIapp}).
\end{algorithm}
\rule{1.05\textwidth}{0.5pt}
\end{minipage}
}
\end{center}

  \subsection{$\text{\sc GWB}_{\text{II}}$ Allocation Method}\label{app:GWBIImethod}
  \begin{center}
  \scalebox{0.7}{
    \begin{minipage}{\linewidth}
\noindent
\rule{1.05\textwidth}{0.5pt}
\noindent
\rule{1.05\textwidth}{0.5pt}
\begin{algorithm}[H]
\begin{center}
\vspace{-.5cm}
\TitleOfAlgo{$\text{\sc GWB}_{\text{II}}$ Allocation Method}
\end{center}

\noindent
{\bf Input}: {$\hvRet$, $\hcov_\Rsc$, $\viewP$, $\nu_\viewR$, $\cov_\viewR$,  $\gammaR$, $\lambda$}

\vspace{.1cm}
{\bf Method}:
\begin{itemize}
 \item Drift Update:
 \be
 \IPPT = \(\eye_\Nprior + \viewP^T \viewP\)^{-1}
 \ee
\be
 \label{eq:GWIImuE}
 \vmm^{(\GBII)}_\Est \leftarrow \vmm_{\GBII} = \IPPT \( \hvRet+ {\lambda} \viewP^T  \vnu_{\viewR} \)
\ee
 \item Covariance Update:
 \be
 A_\Rsc =  \IPPT \hcov_\Rsc \IPPT
 \ee
 \be
 \EuScript{B}_\viewR = \lambda \IPPT A_\Rsc^{-\half} \(A_\Rsc^{\half} \viewP^T \cov_\Rsc \viewP A_\Rsc^{\half}\)^{\half} A_\Rsc^{-\half}  W
 \ee
\be
 \label{eq:GWIIcovE}
 \cov^{(\GBII)}_\Est \leftarrow \cov_\GBII = \(\IPPT+ \EuScript{B}_{\viewR}\)\hcov_\Rsc\(\IPPT+ \EuScript{B}_{\viewR}\)
  \ee

\item {\sc Optimal Weights}: We compute optimal weights with the \GWBI \, update as follows:
\be
\wGWII = \text{{\sc MVO}}\[\vmm^{(\GBII)}_\Est , \cov^{(\GBII)}_\Est; \gammaR\]
\label{eq:GWBIIapp}
\ee
\end{itemize}
{\bf Result}: Weights $\wGWII$ computed in equation (\ref{eq:GWBIIapp}).
\end{algorithm}
\rule{1.05\textwidth}{0.5pt}
\end{minipage}
}
\end{center}

\bibliographystyle{unsrt}

\begin{thebibliography}{99}
\bibitem{BL}
F.~Black {\it \&} R.~Litterman, ``Asset Allocation: Combining Investor Views With Market Equilibrium'', {\it Journal of Fixed Income},  {\bf{2}} (1991)
\bibitem{He}
G.~He  {\it \&}  R.~Litterman, ``The Intution Behind Black-Litterman Model Portfolios'', SSRN: https://ssrn.com/abstract=334304 or http://dx.doi.org/10.2139/ssrn.334304
\bibitem{Bertsimas}
D.~Bertsimas, V.~Gupta {\it \&}  I.~C.Paschalidis, ``A New Perspective on the Black-Litterman Model'', {\it Operations Research}, {\bf{16}} (2012)
\bibitem{Meucci}
A.~Meucci, ``The Black-Litterman Approach: Original Model and Extensions '' SSRN: https://ssrn.com/abstract=1117574  ({\color{blue} \url{ http://dx.doi.org/10.2139/ssrn.1117574}}) (2008)
\bibitem{Rachev}
S.~T.~Rachev, J.~S.~J.~Hsu, B.~S.~Bagasheva, F.~J.~Fabozzi, ``Bayesian Methods in Finance'', Wiley, ISBN: 978-0-470-24924-6 (2008)
\bibitem{Kolm}
P.~Kolm {\it \&} G.~Ritter, ``On the Bayesian interpretation of Black–Litterman'', {\it European Journal of Operational Research}, {\color{blue} \url{http://dx.doi.org/10.1016/j.ejor.2016.10.027}} (2016)
\bibitem{Meucci2014}
A.~Meucci, D.~Ardia, {\it \&} M.~Colasante, ``Portfolio Construction and Systematic Trading with Factor Entropy Pooling'',  Risk Magazine, {\bf{27}}  (2014)
\bibitem{Fabozzi}
F.~J.~Fabozzi, S.~M.~Focardi, {\it \&} P.~N.~Kolm, ``Incorporating Trading Strategies in the Black-Litterman Framework'', {\it The Journal of Trading}, {\bf{1}} (2006)
\bibitem{Duraj}
J.~Duraj,{\it \&} C.~Yu, ``Black-Litterman End-to-End'', {\color{blue} \url{https://ssrn.com/abstract=4532798}} (2023)
\bibitem{GWB}
J.~Delon, N.~Gozlan {\it \&}  A.~Saint-Dizier ``Generalized Wasserstein barycenters between probability measures living on different subspaces'',  {\color{blue} \url{arXiv:2105.09755v1}} (2021)
\bibitem{McCann}
R.~J.~McCann, ``A convexity principle for interacting gases'', {\it Advances in Mathematics}, {\bf{128}} (1997).
\bibitem{WD}
I.~Olkin {\it \&} F.~Pukelsheim, ``The distance between two random vectors with given dispersion matrices'', {\it Linear Algebra and its Applications}, {\bf{48}}  (1982)
\bibitem{Dowson}
D.~C.~Dowson, {\it \&} B.~V.~Landau,, ``The Fréchet distance between multivariate normal distributions'', {\it Journal of Multivariate Analysis}, {\bf{12}} (1982)
\bibitem{Givens}
C.~R.~Givens {\it \&} R.~M.~Shortt, ``A Class of Wasserstein Metrics For Probability Distributions'', {\it Michigan Math Journal}, {\bf{31}} (1984)
\bibitem{Doust} 
P.~Doust, ``Geometric Mean variance'', {\it Risk.net} (2008)
\bibitem{Asuka}
A.~Takatsu, ``Wasserstein geometry of Gaussian measures'', {\it Osaka Journal of Mathematics}, {\bf{48}} 2011.
\bibitem{bhatia} 
R.~Bhatia, T.~Jain {\it \&} Y.~Lim, ``On the Bures-Wasserstein Distance Between Positive Definite Matrices'', {\it Expositiones Mathematicae}, {\bf{37}} (2019)
\bibitem{cvxpy}
S.~Diamond {\it{\&}} ~S.~Boyd, ``CVXPY: A Python-Embedded Modeling Language for Convex Optimization'', {\it Journal of Machine Learning Research} , {\bf{17}} (2016)
\\
A.~Agarwal, R.~Verschueren, ~S.~Diamond {\it{\&}} ~S.~Boyd, ``A rewriting system for convex optimization problems'', {\it Journal of Control and Decision} , {\bf{5}} (2018)
\bibitem{Enzo} 
E.~Busseti, ``Portfolio Management and Optimal Execution via Convex Optimization'',  \href{https://stacks.stanford.edu/file/druid:wm743bj5020/thesis-augmented.pdf}{\color{blue} \it  Ph. D. Thesis}, Stanford University (2018)
\\
 S.~Boyd, E.~Busseti, S.~Diamond, R.~N.~Kahn, K.~Koh, P.~Nystrup, J.~Speth, ``Multi-period trading via convex optimization'', {\color{blue} \url{arXiv: 1705.00109} } (2017)
 \\
\href{ https://www.cvxportfolio.com}{\sc{\color{blue} CVXPortfolio}}:
This website provides the code developed by the group for the above mentioned paper and thesis.
\bibitem{MLDP}
M.~Lopez de Prado, ``Advances in Financial Machine Learning'', John Wiley {\it \&} Sons, Inc. (2018)
\bibitem{mcap}
A.~Petajisto, ``Underperformance of Concentrated Stock Positions'', SSRN: {\color{blue} \url{http://dx.doi.org/10.2139/ssrn.4541122}} (2023)
\bibitem{Tibshirani}
R.~Tibshirani, ``A General Framework for Fast Stagewise Algorithms'', {\it Journal of Machine Learning Research}, {\bf{16}} (2015)
\\
R.~Tibshirani, ``Convex Optimization'', {\it Lectures (CMU)}.
\bibitem{Kloft}
M.~Kloft, U.~Brefeld, S.~Sonnenburg {\it \&} Alexander Zien, ``$\ell_p$-Norm Multiple Kernel Learning'', {\it Journal of Machine Learning Research}, {\bf{12}} (2011)
\bibitem{Bertsekas}
D.~P.~Bertsekas, ``Convex Optimization Theory'', {\it ISBN  9781886529311}, {Athena Scientific} (2009)
\bibitem{BoydII}
S.~Boyd {\it \&} L.~Vandenberghe, ``Convex Optimization'', Cambridge University Press (2004)
\bibitem{Bogachev}
V.~I.~ Bogachev, ``Measure Theory'', {\it ISBN 9783540345138}, {Springer Verlag} (2007)
\end{thebibliography}

\end{document}